\documentclass[journal,onecolumn,12pt]{IEEEtran}
\usepackage[numbers,sort&compress]{natbib}
\usepackage{pslatex}
\usepackage{amsfonts,color,morefloats}
\usepackage{amssymb,amsmath,latexsym,amsthm}
\usepackage{hyperref}
\usepackage{multirow}
\usepackage{makecell}
\hypersetup{hidelinks}
\allowdisplaybreaks

\newcommand{\gf}{{\mathbb{F}}}

\newtheorem{theorem}{Theorem}
\newtheorem{lemma}{Lemma}
\newtheorem{proposition}{Proposition}
\newtheorem{corollary}{Corollary}

\newtheorem{definition}{Definition}
\newtheorem{example}{Example}

\newtheorem{remark}{Remark}

\begin{document}
 	\title{A note on the differential spectrum of the Ness-Helleseth function
	}
	\author{Ketong Ren\IEEEauthorrefmark{1}, Maosheng Xiong\IEEEauthorrefmark{2} and Haode Yan \IEEEauthorrefmark{1}\\
	\IEEEauthorblockA{\IEEEauthorrefmark{1}School of Mathematics, Southwest Jiaotong University, Chengdu, China.\\ E-mail: \href{mailto: rkt@my.swjtu.edu.cn}{rkt}@my.swjtu.edu.cn, \href{mailto: hdyan@swjtu.edu.cn}{hdyan}@swjtu.edu.cn}\\
	\IEEEauthorblockA{\IEEEauthorrefmark{2}Department of Mathematics, The Hong Kong University of Science and Technology, Hong Kong, China.\\ E-mail: mamsxiong@ust.hk}\\
	\emph{Corresponding author: Haode Yan}\\
	}
	%\date{\today} smesnager@univ-paris8.fr
	\maketitle
	\begin{abstract}
		Let $n\geqslant3$ be an odd integer and $u$ an element in the finite field $\gf_{3^n}$. The Ness-Helleseth function is the binomial $f_u(x)=ux^{d_1}+x^{d_2}$ over $\gf_{3^n}$, where $d_1=\frac{3^n-1}{2}-1$ and $d_2=3^n-2$. In 2007, Ness and Helleseth showed that $f_u$ is an APN function when $\chi(u+1)=\chi(u-1)=\chi(u)$, is differentially $3$-uniform when $\chi(u+1)=\chi(u-1)\neq\chi(u)$, and has differential uniformity at most 4 if $ \chi(u+1)\neq\chi(u-1)$ and $u\notin\gf_3$. Here $\chi(\cdot)$ denotes the quadratic character on $\gf_{3^n}$. Recently, Xia et al. determined the differential uniformity of $f_u$ for all $u$ and computed the differential spectrum of $f_u$ for $u$ satisfying $\chi(u+1)=\chi(u-1)$ or $u\in\gf_3$. The remaining problem is the differential spectrum of $f_u$ with $\chi(u+1)\neq\chi(u-1)$ and $u\notin\gf_3$. In this paper, we fill in the gap. By studying differential equations arising from the Ness-Helleseth function $f_u$ more carefully, we express the differential spectrum of $f_u$ for such $u$ in terms of two quadratic character sums. This complements the previous work of Xia et al. 
	\end{abstract}
	
	
	
\section{Introduction}
Substitution boxes (S-boxes for short) are crucial in symmetric block ciphers. Cryptographic functions used in S-boxes can be considered as functions defined over finite fields. Let $\gf_{q}$ be the finite field with $q$ elements, where $q$ is a prime power (i.e. $q=p^n$ and $n$ is a positive integer). We denote by $\gf_{q}^*:=\gf_q \setminus \{0\}$ the multiplicative cyclic subgroup of $\gf_{q}$. Any function $F: \gf_{q} \rightarrow \gf_{q}$ can be uniquely represented as a univariate polynomial of degree less than $q$. For a cryptographic function $F$, the main tools to study $F$ regarding the differential attack \cite{BS91} are the difference distribution table (DDT for short) and the differential uniformity introduced by Nyberg \cite{N94} in 1994. The DDT entry at point $(a,b)$ for any $a,b\in \gf_q$, denoted by $\delta_F(a,b)$, is defined as
\[\delta_F(a,b)=\left|\{x\in\gf_q|\mathbb{D}_aF(x)=b\}\right|,\]
where $\mathbb{D}_a F(x)=F(x+a)-F(x)$ is the \textit{derivative function} of $F$ at the element $a$. The differential uniformity of $F$, denoted by $\Delta_F$, is defined as
\[\Delta_F=\mathrm{max}\left\{\delta_F(a,b)|a\in\gf_q^*, b\in\gf_q\right\}.\]
Generally speaking, the smaller the value of $\Delta_F$, the stronger the resistance of $F$ used in S-boxes against the differential attack. A cryptographic function $F$ is called differentially $k$-uniform if $\Delta_F=k$. Particularly when $\Delta_F=1$, $F$ is called a planar function \cite{DO68} or a perfect nonlinear (abbreviated as PN) function \cite{N91}. When $\Delta_F=2$, $F$ is called an almost perfect nonlinear (abbreviated as APN) function \cite{N94}, which is of the lowest possible differential uniformity over $\gf_{2^n}$ as in such finite fields, no PN functions exist. It has been of great research interest to find new functions with low differential uniformity. Readers may refer to \cite{BCCe20}, \cite{DMMe03}, \cite{HRS99}, \cite{HS97}, \cite{ORG}, \cite{PTW16}, \cite{XCX16}, \cite{ZHS14}, \cite{ZHSS14}, \cite{ZW11} and references therein for some of the new development.

To investigate further differential properties of nonlinear functions, the concept of differential spectrum was devised as a refinement of differential uniformity \cite{BCC10}. 
\begin{definition}\label{def:ds}
Let $F$ be a function from $\gf_{p^n}$ to $\gf_{p^n}$ with differential uniformity $k$, and
$$ \omega_i = |\{(a,b) \in \gf_{p^n}^* \times \gf_{p^n} | \delta_F(a,b) = i\}|, 0 \leqslant i \leqslant k.$$
The differential spectrum of $F$ is defined as the ordered sequence
 $$\mathbb{S} =\left[\omega_0,\omega_1,...,\omega_{k}\right].$$
\end{definition}
According to the Definition \ref{def:ds}, we have the following two identities
\begin{equation}
\sum\limits_{i=0}^{k}\omega_i=(p^n-1)p^n,\sum\limits_{i=0}^{k}i\omega_i=(p^n-1)p^n. \label{eqt:omega}
\end{equation}
The differential spectrum of a cryptographic function, compared with the differential uniformity, provides much more detailed information. In particular, the value distribution of the DDT is given directly by the differential spectrum. Differential spectrum has many applications such as in sequences \cite{BCC05}, \cite{DHKM01}, coding theory \cite{CCZ98}, \cite{CP19}, combinatorial design \cite{TDX20} etc. However, to determine the differential spectrum of a cryptographic function is usually a difficult problem. There are two variables $a$ and $b$ to consider in each $\omega_i$. When $F$ is a power function, i.e., $F(x)=x^d$ for some positive integer $d$, since $\delta_F(a,b)=\delta_F(1,\frac{b}{a^d})$, the problem of the value distribution of $\{\delta_F(a,b)|b\in\gf_q\}$ is the same as that of  $\{\delta_F(1,b)|b\in\gf_q\}$, so in this case two variables $a$ and $b$ degenerate into one variable $b$ and the problem becomes much easier. Power functions with known differential spectra are summarized in Table \ref{tab:DS}.

\begin{table}[!t]
	\footnotesize
	\centering
	\caption{Power Functions over $\gf_{p^n}$ with Known Differential Spectra}
	\begin{tabular}{|c|c|c|c|c|}
		\hline
		$p$ & $d$ & Condition & $\Delta_F$ & Ref \\
		\hline
		
		$2$ & $2^t+1$ & gcd$(t,n)=s$ & $2^s$ & \cite{BCC10} \\
		\hline

		$2$ & $2^{2t}-2^t+1$ & gcd$(t,n)=s,\frac{n}{s}\,odd$ & $2^s$ & \cite{BCC10} \\
		\hline

		$2$ & $2^n-2$ & $n\geqslant2$ & $2\,or\,4$ & \cite{BCC10} \\
		\hline

		$2$ & $2^{2k}+2^k+1$ & $n=4k$ & $4$ & \cite{BCC10},\cite{XY17} \\
		\hline

		$2$ & $2^t-1$ & $t=3,n-2$ & $6$ or $8$ & \cite{BCC11} \\
		\hline

		$2$ & $2^t-1$ & $t=\frac{n-1}{2}$, $\frac{n+3}{2}$, $n$ odd & $6$ or $8$ & \cite{BP14} \\
		\hline

		$2$ & $2^m+2^{(m+1)/2}+1$ & $n=2m$, $m\geqslant5$ odd & $8$ & \cite{XYY18} \\
		\hline

		$2$ & $2^{m+1}+3$ & $n=2m$, $m\geqslant5$ odd & $8$ & \cite{XYY18} \\
		\hline

		$2$ & $2^{3k}+2^{2k}+2^k-1$ & $n=4k$ & $2^{2k}$ & \cite{TuLe23} \\
		\hline
		
		$2$ & $\frac{2^m-1}{2^k+1}+1$ & $n=2m$, gcd$(k,m)=1$ & $2^m$ & \cite{XML23} \\
		\hline

		$3$ & $2\cdot3^{(n-1)/2}+1$ & $n$ odd & $4$ & \cite{DHKM01} \\
		\hline

		$3$ & $\frac{3^n-1}{2}+2$ & $n$ odd & $4$ & \cite{JLLQ22} \\
		\hline

		$5$ & $\frac{5^n-3}{2}$ & any $n$ & $4$ or $5$ & \cite{YL21} \\
		\hline

		$5$ & $\frac{5^n+3}{2}$ & any $n$ & $3$ & \cite{PLZ23} \\
		\hline

		$p$ odd & $p^{2k}-p^k+1$ & gcd$(n,k)=e$, $\frac{n}{e} odd$ & $p^e+1$ & \cite{YZWe19}, \cite{LRF21} \\
		\hline

		$p$ odd & $\frac{p^k+1}{2}$ & gcd$(n,k)=e$ & $\frac{p^e-1}{2}$ or $p^e+1$ & \cite{CHNe13} \\
		\hline

		$p$ odd & $\frac{p^n+1}{p^m+1}+\frac{p^n-1}{2}$ & $p\equiv3\,(mod\;4)$, $m|n$, $n$ odd & $\frac{p^m+1}{2}$ & \cite{CHNe13} \\
		\hline

		$p$ odd & $p^n-3$ & any $n$ & $\leqslant5$ & \cite{XZLH18}, \cite{YXe22} \\
		\hline

		$p$ odd & $p^m+2$ & $n=2m$ & $2$ or $4$ & \cite{HRS99}, \cite{MXLH22} \\
		\hline

		$p$ odd & $2p^{\frac{n}{2}}-1$ & $n$ even & $p^{\frac{n}{2}}$  & \cite{YL22} \\
		\hline

		$p$ odd & $\frac{p^n-3}{2}$ & $p^n\equiv3\,(mod\;4)$, $p^n\geqslant7$ and $p^n\neq27$ & $2$ or $3$ & \cite{YMT24} \\
		\hline

		$p$ odd & $\frac{p^n+3}{2}$ & $p\geqslant5$, $p^n\equiv1\,(mod\;4)$ & $3$ & \cite{JLLQ21} \\
		\hline

		$p$ odd & $\frac{p^n+3}{2}$ & $p^n=11$ or $p^n\equiv3\,(mod\;4)$, $p\neq3$, $p^n\neq11$ & $2$ or $4$ & \cite{YMT23} \\
		\hline

		$p$ odd & $\frac{p^n+1}{4}$ & $p\neq3$, $p^n>7$, $p^n\equiv7\,(mod\;8)$& $2$ & \cite{TY23}, \cite{HRS99} \\
		\hline

		$p$ odd & $\frac{3p^n-1}{4}$ & $p\neq3$, $p^n>7$, $p^n\equiv3\,(mod\;8)$& $2$ & \cite{TY23}, \cite{HRS99} \\
		\hline

		$p$ odd & $\frac{p^n+1}{4}$, $\frac{3p^n-1}{4}$ & $p=3$ or $p>3$, $p^n\equiv3\,(mod\;4)$& $4$ & \cite{BXe23} \\
		\hline
		
		any $p$ & $k(p^m-1)$ & $n=2m$, gcd$(k,p^m+1)=1$ & $p^m-2$ & \cite{HLXe23} \\
		\hline
	\end{tabular}
	\label{tab:DS}
\end{table}

For a polynomial that is not a power function, the investigation of its differential spectrum is much more difficult. There are very few such functions whose differential spectra were known \cite{LgJ24}, \cite{PdSW17}. The main focus of this paper is the Ness-Helleseth function. Let $n$ be a positive odd integer, $d_1=\frac{3^n-1}{2}-1$, $d_2=3^n-2$ and $u\in\gf_{3^n}$. The \textit{Ness-Helleseth function}, denoted as $f_u(x)$, is a binomial over $\gf_{3^n}$ defined as
\begin{equation}
	f_u(x)=ux^{d_1}+x^{d_2}. \label{NHf}
\end{equation}
To describe the differential properties of the Ness-Helleseth function $f_u(x)$ which obviously depend on $u$, we define certain sets of $u$ as in \cite{FurXia} 
$$
	\left\{
            \begin{array}{lcl}
			\mathcal{U}_0=\left\{u\in\gf_{3^n} | \chi(u+1)\neq\chi(u-1)\right\}, \\
            \mathcal{U}_1=\left\{u\in\gf_{3^n} | \chi(u+1)=\chi(u-1)\right\}, \\
			\mathcal{U}_{10}=\left\{u\in\gf_{3^n} | \chi(u+1)=\chi(u-1)\neq\chi(u)\right\}, \\
			\mathcal{U}_{11}=\left\{u\in\gf_{3^n} | \chi(u+1)=\chi(u-1)=\chi(u)\right\}.
            \end{array}
	\right.
$$
Here $\chi$ denotes the quadratic character on $\gf_{3^n}^*$. It is easy to see that $\mathcal{U}_0\cap\mathcal{U}_1=\emptyset$, $\mathcal{U}_{10}\cap\mathcal{U}_{11}=\emptyset$ and $\mathcal{U}_{10}\cup\mathcal{U}_{11}=\mathcal{U}_1$.

In 2007, Ness and Helleseth showed that (see \cite{ORG})
\begin{itemize}
\item[1).] $f_u$ is an APN function when $u\in\mathcal{U}_{11}$; 

\item[2).] $f_u$ is differentially $3$-uniform when $u\in \mathcal{U}_{10}$;  

\item[3).] $f_u$ has differential uniformity at most 4 if $u\in\mathcal{U}_0\setminus \gf_3$. 
\end{itemize}
Moreover, Ness and Helleseth observed by numerical computation that in 1), the constraint imposed on $u$, namely $u\in \mathcal{U}_{11}$, appears to be necessary for $f_u$ to be an APN function. 

In a recent paper \cite{FurXia}, Xia et al. conducted a further investigation into the differential properties of the Ness-Helleseth function $f_u$. They determined the differential uniformity of $f_u$ for all $u \in \gf_{3^n}$ (see \cite[Theorem 4]{FurXia}), hence confirming, in particular, that $f_u$ is indeed APN if and only if $u\in\mathcal{U}_{11}$. Moreover, for the cases of 1) and 2), they also computed the differential spectrum of $f_u$ explicitly in terms of a quadratic character sum $T(u)$ (see \cite[Propositions 5 and 6]{FurXia}). However, for $u\in\mathcal{U}_0\setminus \gf_3$, while it was shown that $f_u$ has differential uniformity 4, the differential spectrum of $f_u$ remains open. The purpose of this paper is to fill in this gap, that is, in this paper, we will compute the differential spectrum of $f_u$ explicitly for any $u\in\mathcal{U}_0\setminus \gf_3$ and similar to \cite[Propositions 5 and 6]{FurXia}, the result will be expressed in terms of quadratic character sums depending on $u$. 

Let us make a comparison of the methods used in this paper and in \cite{FurXia}. We first remark that for $u\in\mathcal{U}_0\setminus \gf_3$, to determine the differential uniformity of $f_u$ is already a quite difficult problem, as was shown in \cite{FurXia}, the final result involved 32 different quadratic character sums, about one-half of which can not be evaluated easily (see \cite[Table V]{FurXia}). Instead, the authors applied Weil's bound on many of these character sums over finite fields to conclude that the differential uniformity of $f_u$ is 4. While our paper is based on \cite{FurXia} and can be considered as a refinement, since we are dealing with the differential spectrum which is a much more difficult problem, it is conceivable that the techniques involved in this paper would be even more complicated. This is indeed the case, as will be seen in the proofs later on. In particular, we have found many relations among these 32 character sums, some of which are quite technical and surprising, that help up in computing the differential spectrum.

%\cite{LYU2024102453}

This paper is organized as follows. Section \ref{sec:CS} presents certain quadratic character sums that are essential for the computation of the differential spectrum. In Section \ref{sec:DSpre}, the necessary and sufficient conditions of the differential equation to have $i~ (i=0,1,2,3,4)$ solutions are given. The differential spectrum of $f_u$ is investigated in Section \ref{sec:DS}. Section \ref{sec:con} concludes this paper.

\section{On quadratic character sums}\label{sec:CS}
In this section, we will introduce some results on the quadratic character sum over finite fields. Let $\chi(\cdot)$ be the quadratic character of $\gf_{p^n}$ ($p$ is an odd prime), which is defined as
$$\chi(x)=
\left\{
\begin{array}{cl}
	1, & \mbox{if } x \mbox{ is a square in } \gf_{p^n}^*,\\
	-1,& \mbox{if } x \mbox{ is a nonsquare in }\gf_{p^n}^*, \\
	0, & \mbox{if } x=0.
\end{array}
\right.
$$
Let $\gf_{p^n}[x]$ be the polynomial ring over $\gf_{p^n}$. We consider the character sum of the form
\begin{eqnarray} \label{2:x1} \sum_{a\in\gf_{p^n}}\chi(f(a))\end{eqnarray}
with $f\in\gf_{p^n}[x]$. The case of $\deg(f) =1$ is trivial, and for $\deg(f) =2$, the following explicit formula was established in \cite{FF}.
\begin{lemma}\cite[Theorem 5.48]{FF}
	Let $f(x)=a_2x^2+a_1x+a_0\in\gf_q[x]$ with $q$ odd and $a_2\neq 0$. Put $d=a_1^2-4a_0a_2$ and let $\chi(\cdot)$ be the quadratic character of $\gf_q$. Then
	$$\sum\limits_{a\in\gf_q}\chi(f(a))=
	\left\{
	\begin{array}{lcl}
		-\chi(a_2), & if\;d\neq 0,\\
		(q-1)\chi(a_2), & if\;d=0. \\
	\end{array}
	\right.
	$$
\end{lemma}

The character sum plays an important role in determining the differential spectrum of the Ness-Helleseth function. Let $p=3$ and $n$ be an odd integer. For any fixed $u \in\mathcal{U}_0\setminus \gf_3$, define $g_i\in\gf_{3^n}[x]$ ($i\in\{1,2,3,4,5\}$) as follows: 
$$
\left\{
\begin{array}{ll}
	g_1(x)=-(u+1)x, \\
	g_2(x)=x(x-1-u), \\
	g_3(x)=x(x-1+u), \\
	g_4(x)=x^2-x+u^2=(x+1+\sqrt{1-u^2})(x+1-\sqrt{1-u^2}), \\
	 g_5(x)=-\varphi(u)(x+\frac{u^2}{\varphi(u)})=-\varphi(u)(x+1-\sqrt{1-u^2}),\;where\;\varphi(u)=1+\sqrt{1-u^2}.\\
\end{array}
\right.
$$
Herein and hereafter, for a square $x \in \gf_{3^n}^*$, we denote by $\sqrt{x}$ the square root of $x$ in $\gf_{3^n}$ such that $\chi(\sqrt{x})=1$. Since $n$ is odd, $\chi(-1)=-1$, this $\sqrt{x}$ is uniquely determined by $x$. For $u \in \mathcal{U}_0\setminus \gf_3$, the element $1-u^2$ is always a square in $\gf_{3^n}$, so $\sqrt{1-u^2}$ is well defined, and we have $\chi((u+1)\varphi(u))=\chi(-(u+1+\sqrt{1-u^2})^2)=-1$. Additionally, let $A=\{0,1\pm u,-1\pm \sqrt{1-u^2}\}$. It is easy to note that the set $A$ contains all the zeros of $g_i(x)$, $i=1,2,3,4,5$. The values of $\chi(g_i(x))$ on $A$ are displayed in the Table \ref{tab:vA}.
	\begin{table}[!htp]
		\footnotesize
		\centering
		\caption{The values of $\chi(g_i(x))$ on set A}
		\begin{tabular}{|c|c|c|c|c|c|}
			\hline
			$x$ & $\chi(g_1(x))$ & $\chi(g_2(x))$ & $\chi(g_3(x))$ & $\chi(g_4(x))$ & $\chi(g_5(x))$ \\
			\hline
			0 & 0 & 0  & 0 & 1 & $-1$ \\
			\hline
			$1+u$ & $-1$ & 0 & $-\chi(u^2+u)$ & $\chi(u-u^2)$ & $-\chi((u+1)(\sqrt{1-u^2})+(u-1)^2)$\\
			\hline
			$1-u$ & $-1$ & $\chi(u-u^2)$ & 0 & $-\chi(u^2+u)$ & $-\chi((1-u)(\sqrt{1-u^2})+(u+1)^2)$ \\
			\hline
			$-1+\sqrt{1-u^2}$ & $-1$ & $-\chi(u)\chi(-1+u+\sqrt{1-u^2})$ & $\chi(u)\chi(-1-u+\sqrt{1-u^2})$  & 0 & 0 \\
			\hline
			$-1-\sqrt{1-u^2}$ & $-1$ & $\chi(u)\chi(1-u+\sqrt{1-u^2})$ & $-\chi(u)\chi(1+u+\sqrt{1-u^2})$ & 0 & $\chi(u^2-1-\sqrt{1-u^2})$ \\
			\hline
		\end{tabular}
		\label{tab:vA}
	\end{table}
	
For any $u \in \mathcal{U}_0\backslash\gf_3$, the following character sums were meticulously computed in \cite{FurXia}:	
\begin{lemma}\cite{FurXia}
Let $u \in \mathcal{U}_0\backslash\gf_3$, we have 
\begin{itemize}
	\item $\sum\limits_{z\in\gf_{3^n}}\chi(g_1(z)g_2(z))=-1,$ 
	\item $\sum\limits_{z\in\gf_{3^n}}\chi(g_1(z)g_3(z))=-1,$ 
	\item $\sum\limits_{z\in\gf_{3^n}}\chi(g_1(z)g_5(z))=1,$ 
	\item $\sum\limits_{z\in\gf_{3^n}}\chi(g_2(z)g_3(z))=-2.$
	\item $\sum\limits_{z\in\gf_{3^n}}\chi(g_1(z)g_2(z)g_5(z))=2, $ 
	\item $\sum\limits_{z\in\gf_{3^n}}\chi(g_1(z)g_3(z)g_5(z))=2.$
\end{itemize}

\end{lemma}

In what follows, we give a series of lemmas on quadratic character sums involving $g_i(x)$. The first three lemmas can be proved directly.
\begin{lemma}When $u \in \mathcal{U}_0\backslash\gf_3$, we have
	\[ \sum_{z\in\gf_{3^n}}\chi(g_4(z)g_5(z))=-\chi(\varphi(u)).\]
\end{lemma}
\begin{proof}We have,
	\begin{align*}
		\sum_{z\in\gf_{3^n}}\chi(g_4(z)g_5(z))&=-\chi(\varphi(u))\sum_{z\in\gf_{3^n}} \chi((z+1-\sqrt{1-u^2})^2(z+1+\sqrt{1-u^2})) \\
		&=-\chi(\varphi(u))\sum_{z\in\gf_{3^n},z\neq \sqrt{1-u^2}-1}\chi(z+1+\sqrt{1-u^2}) \\
		&=-\chi(\varphi(u))\chi(\sqrt{1-u^2})
		=-\chi(\varphi(u)).
	\end{align*}
\end{proof}
\begin{lemma}When $u \in \mathcal{U}_0\backslash\gf_3$, we have
	\[ \sum_{z\in\gf_{3^n}}\chi(g_1(z)g_4(z)g_5(z))=1+\chi(\varphi(u)).\]
\end{lemma}

\begin{proof} We have,
	\begin{align*}
		\sum_{z\in\gf_{3^n}}\chi(g_1(z)g_4(z)g_5(z))&=\chi(u+1)\chi(\varphi(u))\sum_{z\in\gf_{3^n}} \chi( z(z+1+\sqrt{1-u^2})(z+1-\sqrt{1-u^2})^2) \\
		&=-\sum_{z\in\gf_{3^n},z\neq \sqrt{1-u^2}-1} \chi( z(z+1+\sqrt{1-u^2})) \\
		&=-(-1-\chi((\sqrt{1-u^2}-1)(-\sqrt{1-u^2})))\\
		&=1-\chi(\sqrt{1-u^2}-1)
		=1+\chi(\varphi(u)).
	\end{align*}
\end{proof}

\begin{lemma}When $u \in \mathcal{U}_0\backslash\gf_3$, we have
	\[\sum_{z\in\gf_{3^n}}\chi(g_1(z)g_3(z)g_4(z)g_5(z))=2-\chi((\sqrt{1-u^2}+1+u)),\]
	and
	\[\sum_{z\in\gf_{3^n}}\chi(g_1(z)g_2(z)g_4(z)g_5(z))=2-\chi((\sqrt{1-u^2}+1-u)).\]
\end{lemma}
\begin{proof}We only prove the first identity, as the proof of the second is very similar. 
	\begin{align*}
		 \sum_{z\in\gf_{3^n}}\chi(g_1(z)g_3(z)g_4(z)g_5(z))=&\chi((u+1)\varphi(u))\sum_{z\in\gf_{3^n}}\chi((z-1+u)(z+1+\sqrt{1-u^2})z^2(z+1-\sqrt{1-u^2})^2)\\
		=&-\sum_{z\in\gf_{3^n}^*,z\neq\sqrt{1-u^2}-1}\chi((z-1+u)(z+1+\sqrt{1-u^2}))\\
		=&-\sum_{z\in\gf_{3^n}}\chi((z-1+u)(z+1+\sqrt{1-u^2}))\\
		&+\chi((u-1)\varphi(u))+\chi((\sqrt{1-u^2}+1+u)(-\sqrt{1-u^2}))\\
		=&1-\chi((u+1)\varphi(u))-\chi((\sqrt{1-u^2}+1+u))\\
		=&2-\chi((\sqrt{1-u^2}+1+u)).
	\end{align*}
%	The second identity can be obtained similarly.
\end{proof}

\begin{lemma}When $u \in \mathcal{U}_0\backslash\gf_3$, we have
	\[\sum_{z\in\gf_{3^n}}\chi(g_2(z)g_3(z)g_4(z))=-2.\]
\end{lemma}

\begin{proof}We have,
	\begin{align*}
		\sum_{z\in\gf_{3^n}}\chi(g_2(z)g_3(z)g_4(z)) &=\sum_{z\in\gf_{3^n}}\chi(z^2(z-1+u)(z-1-u)(z^2-z+u^2)) \\
		&=\sum_{z\in\gf_{3^n}^*}\chi((z-1+u)(z-1-u)(z^2-z+u^2)) \\
		&=\sum_{z\in\gf_{3^n}^*, z\neq 1\pm u}\chi(\frac{z^2-z+u^2}{z^2+z+1-u^2}) \\
		&=\sum_{z\in\gf_{3^n}, z\neq 1 \pm u}\chi(\frac{z^2-z+u^2}{z^2+z+1-u^2})-\chi(\frac{u^2}{1-u^2}).
	\end{align*}
	Let $t=\frac{z^2-z+u^2}{z^2+z+1-u^2}$. Then $(t-1)z^2+(t+1)z+t(1-u^2)-u^2=0$. We know that $t=1$ if and only if $z=1+u^2$. When $t\neq 1$, the discriminant of the quadratic equation on $z$ is $\Delta_t=(t+1)^2-(t-1)(t(1-u^2)-u^2)=u^2t^2+1-u^2$. The number of $z$ with a fixed $t\neq 1$ is $1+\chi(\Delta_t)$. Hence, we have
	\begin{align*}
		\sum_{z\in\gf_{3^n}}\chi(g_2(z)g_3(z)g_4(z)) &=\sum_{z\in\gf_{3^n}, z\neq 1 \pm u}\chi(\frac{z^2-z+u^2}{z^2+z+1-u^2})-\chi(\frac{u^2}{1-u^2}) \\
		&=\sum_{t\in\gf_{3^n}, t\neq 1}\chi(t)(1+\chi(\Delta_t))-\chi(\frac{u^2}{1-u^2}) \\
		&=\sum_{t\in\gf_{3^n}}\chi(t)(1+\chi(\Delta_t))-\chi(1)-\chi(\frac{u^2}{1-u^2}) \\
		 &=\sum_{t\in\gf_{3^n}}\chi(t)+\sum_{t\in\gf_{3^n}}\chi(t(u^2t^2+1-u^2))-\chi(1)-\chi(\frac{u^2}{1-u^2}).
	\end{align*}
	Note that $\sum_{t\in\gf_{3^n}}\chi(t(u^2t^2+1-u^2))=\sum_{t\in\gf_{3^n}}\chi(-t(u^2t^2+1-u^2))=-\sum_{t\in\gf_{3^n}}\chi(t(u^2t^2+1-u^2))$, then $\sum_{t\in\gf_{3^n}}\chi(t(u^2t^2+1-u^2))=0$. This with $\chi(\frac{u^2}{1-u^2})=1$ leads to $ \sum_{z\in\gf_{3^n}}\chi(g_2(z)g_3(z)g_4(z))=-2$.
\end{proof}

\begin{lemma}When $u \in \mathcal{U}_0\backslash\gf_3$, we have
	\[ \sum_{z\in\gf_{3^n}}\chi(g_1(z)g_4(z))+\sum_{z\in\gf_{3^n}}\chi(g_1(z)g_2(z)g_3(z))=0.\]
\end{lemma}
\begin{proof}First,
	\[\sum_{z\in\gf_{3^n}}\chi(g_1(z)g_4(z))=-\chi(u+1)\sum_{z\in\gf_{3^n}}\chi(z(z^2-z+u^2)),\]
	and
	\begin{align*}
		\sum_{z\in\gf_{3^n}}\chi(g_1(z)g_2(z)g_3(z)) &=-\chi(u+1)\sum_{z\in\gf_{3^n}}\chi(z^3(z-(u+1))(z+(u-1)))\\
		&=-\chi(u+1)\sum_{z\in\gf_{3^n}^*}\chi(z(z^2+z+1-u^2)).
	\end{align*}
	Note that
	 \[\sum_{z\in\gf_{3^n}^*}\chi(z(z^2+z+1-u^2))=-\sum_{z\in\gf_{3^n}^*}\chi(z(z^2-z+1-u^2))=-\sum_{z\in\gf_{3^n}^*}\chi(\frac{z^2-z+1-u^2}{z}).\]
	Let $\frac{z^2-z+1-u^2}{z}=t$. Then $t$ satisfies the quadratic equation \[z^2-(t+1)z+1-u^2=0.\] Clearly, $z=0$ is not the solution of this quadratic equation for any $t\in\gf_{3^n}$ since $u\notin \gf_3$. For each $t$, the number of solutions of $z$ is $1+\chi(\Delta_t)$, where $\Delta_t=(t+1)^2-(1-u^2)=t^2-t+u^2$. Hence
	 \[\sum_{z\in\gf_{3^n}^*}\chi(\frac{z^2-z+1-u^2}{z})=\sum_{t\in\gf_{3^n}}\chi(t)(1+\chi(\Delta_t))=\sum_{t\in\gf_{3^n}}\chi(t\Delta_t)=\sum_{t\in\gf_{3^n}}\chi(t(t^2-t+u^2)).\]
	The desired result follows.
\end{proof}

\begin{lemma}When $u \in \mathcal{U}_0\backslash\gf_3$, we have
	\[ \sum_{z\in\gf_{3^n}}\chi(g_2(z)g_4(z))+\sum_{z\in\gf_{3^n}}\chi(g_1(z)g_2(z)g_4(z))=-2,\]
	and
		\[ \sum_{z\in\gf_{3^n}}\chi(g_3(z)g_4(z))+\sum_{z\in\gf_{3^n}}\chi(g_1(z)g_3(z)g_4(z))=-2.\]
\end{lemma}
\begin{proof}We only prove the first identity. The proof of the second one is similar, so we omit it.
	Note that
	\[\sum_{z\in\gf_{3^n}}\chi(g_2(z)g_4(z))=\sum_{z\in\gf_{3^n}}\chi(z(z-1-u)(z^2-z+u^2)),\]
	and
	 \[\sum_{z\in\gf_{3^n}}\chi(g_1(z)g_2(z)g_4(z))=-\chi(u+1)\sum_{z\in\gf_{3^n}}\chi(z^2(z-1-u)(z^2-z+u^2)).\]
	We have
	\begin{align*}
		\sum_{z\in\gf_{3^n}}\chi(g_1(z)g_2(z)g_4(z)) &=-\chi(u+1)\sum_{z\in\gf_{3^n}^*}\chi((z-1-u)(z^2-z+u^2)) \\
		&=-\chi(u+1)\sum_{z\in\gf_{3^n}, z\ne -1-u}\chi(z((z+1+u)^2-(z+1+u)+u^2)) \\
		&=-\chi(u+1)\sum_{z\in\gf_{3^n}, z\ne -1-u}\chi(z(z^2-(u-1)z-(u^2-u))) \\
		&=-\chi(u+1)(\sum_{z\in\gf_{3^n}}\chi(z(z^2-(u-1)z-(u^2-u))-\chi(-(u+1)u^2))) \\
		&=-1-\chi(u+1)\sum_{z\in\gf_{3^n}}\chi(z(z^2-(u-1)z-(u^2-u))), \\
	\end{align*}
	and 	
	 \[\sum_{z\in\gf_{3^n}}\chi(g_2(z)g_4(z))=\sum_{z\in\gf_{3^n}}\chi(z(z-1-u)(z^2-z+u^2))=\sum_{z\in\gf_{3^n}^*, z\ne u+1}\chi(\frac{z^2-z+u^2}{z(z-1-u)}).\]
	Let $\frac{z^2-z+u^2}{z(z-1-u)}=t$. Then $t$ satisfies
	\[(t-1)z^2+(1-(u+1)t)z-u^2=0.\]
	We know that $t=1$ if and only if $z=-u$. When $t\neq 1$, the discriminant of the quadratic equation on $z$ is $\Delta_t=(1-(u+1)t)^2+(t-1)u^2=(u+1)^2t^2+(u-1)^2t+1-u^2$. The number of $z$ with a fixed $t\neq 1$ is $1+\chi(\Delta_t)$. Hence, we have
	\begin{align*}
		\sum_{z\in\gf_{3^n}}\chi(g_2(z)g_4(z))&=\sum_{z\in\gf_{3^n}^*, z\ne u+1}\chi(\frac{z^2-z+u^2}{z(z-1-u)}) \\
		&=1+\sum_{t\in\gf_{3^n}, t\ne 1}\chi(t)(1+\chi(\Delta_t)) \\
		&=1+\sum_{t\in\gf_{3^n}}\chi(t)(1+\chi(\Delta_t))-(1+\chi(u^2)) \\
		&=\sum_{t\in\gf_{3^n}}\chi(t)(1+\chi(\Delta_t))-1. \\
	\end{align*}
	Note that $\sum\limits_{t\in\gf_{3^n}}\chi(t\Delta_t)=\sum\limits_{t\in\gf_{3^n}}\chi(t((u+1)^2t^2+(u-1)^2t+1-u^2))=\sum\limits_{t\in\gf_{3^n}^*}\chi(\frac{(u+1)^2t^2+(u-1)^2t+1-u^2}{t})$. Let $v=\frac{(u+1)^2t^2+(u-1)^2t+1-u^2}{t}$. Then
	\[(u+1)^2t^2+((u-1)^2-v)t+(1-u^2)=0,\]
	which is a quadratic equation on $t$. $\Delta_v=((u-1)^2-v)^2-(u+1)^2(1-u^2)=v^2+(u-1)^2v-(u-1)u^3.$
	Then
	\begin{align*}
		\sum_{t\in\gf_{3^n}}\chi(t\Delta_t)&=\sum_{v\in\gf_{3^n}}\chi(v(1+\chi(\Delta_v)))\\
		&=\sum_{v\in\gf_{3^n}}\chi(v(v^2+(u-1)^2v-(u-1)u^3)) \\
		&=\sum_{w\in\gf_{3^n}}\chi((u-1)^2w((u-1)^4w^2+(u-1)^4w-(u-1)u^3)) \\
		&=\sum_{w\in\gf_{3^n}}\chi(w(w^2+w-\frac{u^3}{(u-1)^3})) \\
		&=\sum_{w\in\gf_{3^n}}\chi(w^{1/3}((w^{1/3})^2+w^{1/3}-\frac{u^3}{(u-1)^3})) \\
		&=\sum_{w\in\gf_{3^n}}\chi(w(w^2+w-\frac{u}{u-1})) \\
		 &=\sum_{w\in\gf_{3^n}}\chi(\frac{w}{u-1}((\frac{w}{u-1})^2+\frac{w}{u-1}-\frac{u}{u-1})) \\
		&=\chi(u-1)\sum_{w\in\gf_{3^n}}\chi(w(w^2+(u-1)w-(u^2-u))) \\
		&=-\chi(u-1)\sum_{w\in\gf_{3^n}}\chi(w(w^2-(u-1)w-(u^2-u))).
	\end{align*}
	We conclude that
	 \[\sum_{z\in\gf_{3^n}}\chi(g_1(z)g_2(z)g_4(z))=-1-\chi(u+1)\sum_{z\in\gf_{3^n}}\chi(z(z^2-(u-1)z-(u^2-u))),\]
	and
	 \[\sum_{z\in\gf_{3^n}}\chi(g_2(z)g_4(z))=-1-\chi(u-1)\sum_{w\in\gf_{3^n}}\chi(w(w^2-(u-1)w-(u^2-u))).\]
	Then we have
	\[\sum_{z\in\gf_{3^n}}\chi(g_2(z)g_4(z))+\sum_{z\in\gf_{3^n}}\chi(g_1(z)g_2(z)g_4(z))=-2.\]
\end{proof}
\begin{lemma}When $u \in \mathcal{U}_0\backslash\gf_3$, we have
	\[ \sum_{z\in\gf_{3^n}}\chi(g_2(z)g_3(z)g_5(z))+\sum_{z\in\gf_{3^n}}\chi(g_1(z)g_2(z)g_3(z)g_5(z))=2.\]
\end{lemma}

\begin{proof}Note that
	\begin{align*}
		&\sum_{z\in\gf_{3^n}}\chi(g_2(z)g_3(z)g_5(z))\\
		&=-\chi(\varphi(u))\sum_{z\in\gf_{3^n}}\chi(z^2(z-1-u)(z-1+u)(z+1-\sqrt{1-u^2})) \\
		&=-\chi(\varphi(u))\sum_{z\in\gf_{3^n}^*}\chi((z-1-u)(z-1+u)(z+1-\sqrt{1-u^2})) \\
		 &=-\chi(\varphi(u))(-\chi((-1-u)(-1+u)(1-\sqrt{1-u^2}))+\sum_{z\in\gf_{3^n}}\chi((z-1-u)(z-1+u)(z+1-\sqrt{1-u^2})))\\ 	 
		&=1-\chi(\varphi(u))\sum_{z\in\gf_{3^n}}\chi((z-1-u)(z-1+u)(z+1-\sqrt{1-u^2})),
	\end{align*}
		and
	\begin{align*}
		&\sum_{z\in\gf_{3^n}}\chi((z-1-u)(z-1+u)(z+1-\sqrt{1-u^2})) \\
		=&\sum_{z\in\gf_{3^n}}\chi((z+\sqrt{1-u^2}+1-u)(z+\sqrt{1-u^2}+1+u)z) \\
		=&\sum_{z\in\gf_{3^n}}\chi(z(z^2-(\sqrt{1-u^2}+1)z+(u^2-1-\sqrt{1-u^2}))) \\
		 =&\chi(\sqrt{1-u^2}+1)\sum_{z\in\gf_{3^n}}\chi(z(z^2-z+\frac{-u^2+1-\sqrt{1-u^2}}{u^2})).
	\end{align*}
	Then
	 \[\sum_{z\in\gf_{3^n}}\chi(g_2(z)g_3(z)g_5(z))=1-\sum_{z\in\gf_{3^n}}\chi(z(z^2-z+\frac{-u^2+1-\sqrt{1-u^2}}{u^2})).\]			 
	Moreover,
	\begin{align*}
		 &\sum_{z\in\gf_{3^n}}\chi(g_1(z)g_2(z)g_3(z)g_5(z)) \\
		 =&\chi(u+1)\chi(\varphi(u))\sum_{z\in\gf_{3^n}}\chi(z^3(z-1-u)(z-1+u)(z+1-\sqrt{1-u^2})) \\
		=&-\sum_{z\in\gf_{3^n}}\chi(z(z-1-u)(z-1+u)(z+1-\sqrt{1-u^2})) \\
		 =&-\sum_{t\in\gf_{3^n}^*}\chi(\frac{1}{t}(\frac{1}{t}-1-u)(\frac{1}{t}-1+u)(\frac{1}{t}+1-\sqrt{1-u^2})) \\
		=&-\sum_{t\in\gf_{3^n}^*}\chi((1-(1+u)t)(1-(1-u)t)(1+(1-\sqrt{1-u^2})t)) \\
		=&1-\sum_{t\in\gf_{3^n}}\chi((1-(1+u)t)(1-(1-u)t)(1+(1-\sqrt{1-u^2})t)) \\
		 =&1-\chi((1+u)(1-u)(1-\sqrt{1-u^2}))\sum_{t\in\gf_{3^n}}\chi((\frac{1}{1+u}-t)(\frac{1}{1-u}-t)(\frac{1}{1-\sqrt{1-u^2}}+t)) \\
		 =&1-\chi(1-\sqrt{1-u^2})\sum_{t\in\gf_{3^n}}\chi((t-\frac{1}{1+u})(t-\frac{1}{1-u})(t+\frac{1}{1-\sqrt{1-u^2}})) \\
		 =&1-\chi(1-\sqrt{1-u^2})\sum_{t\in\gf_{3^n}}\chi((t-\frac{1}{1-\sqrt{1-u^2}}-\frac{1}{1+u})(t-\frac{1}{1-\sqrt{1-u^2}}-\frac{1}{1-u})t) \\
		 =&1-\chi(1-\sqrt{1-u^2})\sum_{t\in\gf_{3^n}}\chi(t(t^2+\frac{1+(1-u^2)\sqrt{1-u^2}}{u^2(1-u^2)}t-\frac{u^4+u^2+1+\sqrt{1-u^2}}{u^4(1-u^2)})) \\
		=&1-\chi(1-\sqrt{1-u^2})\chi(\frac{1+(1-u^2)\sqrt{1-u^2}}{u^2(1-u^2)})\\
		 &\cdot\sum_{t\in\gf_{3^n}}\chi(t(t^2+t-\frac{u^4+u^2+1+\sqrt{1-u^2}}{u^4(1-u^2)}/(\frac{1+(1-u^2)\sqrt{1-u^2}}{u^2(1-u^2)})^2))\\
		 =&1-\sum_{t\in\gf_{3^n}}\chi(t(t^2+t-\frac{u^4+u^2+1+\sqrt{1-u^2}}{u^4(1-u^2)}/(\frac{1+(1-u^2)\sqrt{1-u^2}}{u^2(1-u^2)})^2))\\
		=&1-\sum_{t\in\gf_{3^n}}\chi(t(t^2+t+\frac{-(u^2-1)^3-(\sqrt{1-u^2})^3}{u^6}))\\
		=&1-\sum_{t\in\gf_{3^n}}\chi(t(t^2+t+\frac{-u^2+1-\sqrt{1-u^2}}{u^2}))\\
		=&1+\sum_{t\in\gf_{3^n}}\chi(t(t^2-t+\frac{-u^2+1-\sqrt{1-u^2}}{u^2})).
	\end{align*}
	The desired result follows.
\end{proof}

\begin{lemma}When $u \in \mathcal{U}_0\backslash\gf_3$, we have
	\[ \sum_{z\in\gf_{3^n}}\chi(g_2(z)g_3(z)g_4(z)g_5(z))+\sum_{z\in\gf_{3^n}}\chi(g_1(z)g_2(z)g_3(z)g_4(z)g_5(z))=2.\]
\end{lemma}

\begin{proof}We have
	\begin{align*}
		&\sum_{z\in\gf_{3^n}}\chi(g_2(z)g_3(z)g_4(z)g_5(z)) \\
		 =&-\chi(\varphi(u))\sum_{z\in\gf_{3^n}}\chi(z^2(z-1-u)(z-1+u)(z+1+\sqrt{1-u^2})(z+1-\sqrt{1-u^2})^2) \\
		=&-\chi(\varphi(u))\sum_{z\in\gf_{3^n}, z\ne 0, z\ne \sqrt{1-u^2}-1}\chi((z-1-u)(z-1+u)(z+1+\sqrt{1-u^2})) \\
		 =&-\chi(\varphi(u))(-\chi((-1-u)(-1+u)(1+\sqrt{1-u^2}))-\chi((\sqrt{1-u^2}+1-u)(\sqrt{1-u^2}+1+u)(-\sqrt{1-u^2})) \\
		&+\sum_{z\in\gf_{3^n}}\chi((z-1-u)(z-1+u)(z+1+\sqrt{1-u^2})))\\
		=&2-\chi(\varphi(u))\sum_{z\in\gf_{3^n}}\chi((z-1-u)(z-1+u)(z+1+\sqrt{1-u^2})).
	\end{align*}
	Note that
	\begin{align*}
		&\sum_{z\in\gf_{3^n}}\chi((z-1-u)(z-1+u)(z+1+\sqrt{1-u^2})) \\
		=&\sum_{z\in\gf_{3^n}}\chi(z(z+1-\sqrt{1-u^2}-u)(z+1-\sqrt{1-u^2}+u))\\
		=&\sum_{z\in\gf_{3^n}}\chi(z(z^2+(\sqrt{1-u^2}-1)z+(u^2-1+\sqrt{1-u^2})))\\
		 =&\chi(\sqrt{1-u^2}-1)\sum_{z\in\gf_{3^n}}\chi(z(z^2+z+\frac{1-u^2+\sqrt{1-u^2}}{u^2})).
	\end{align*}
	
	Then
	 \[\sum_{z\in\gf_{3^n}}\chi(g_2(z)g_3(z)g_4(z)g_5(z))=2+\sum_{z\in\gf_{3^n}}\chi(z(z^2+z+\frac{1-u^2+\sqrt{1-u^2}}{u^2})).\]
	Moreover,
	\begin{align*}
		 &\sum_{z\in\gf_{3^n}}\chi(g_1(z)g_2(z)g_3(z)g_4(z)g_5(z))\\
		 &=\chi(\varphi(u))\chi(u+1)\sum_{z\in\gf_{3^n}}\chi(z^3(z-1-u)(z-1+u)(z+1+\sqrt{1-u^2})(z+1-\sqrt{1-u^2})^2)\\
		&=-\sum_{z\in\gf_{3^n},z\ne \sqrt{1-u^2}-1}\chi(z(z-1-u)(z-1+u)(z+1+\sqrt{1-u^2})) \\
		&=-(-\chi((\sqrt{1-u^2}-1)(\sqrt{1-u^2}+1-u)(\sqrt{1-u^2}+1+u)(-\sqrt{1-u^2})) \\
		&\quad+\sum_{z\in\gf_{3^n}}\chi(z(z-1-u)(z-1+u)(z+1+\sqrt{1-u^2}))) \\
		&=-1-\sum_{z\in\gf_{3^n}}\chi(z(z-1-u)(z-1+u)(z+1+\sqrt{1-u^2})).\\
	\end{align*}
	
	Note that
	\begin{align*}
		&\sum_{z\in\gf_{3^n}}\chi(z(z-1-u)(z-1+u)(z+1+\sqrt{1-u^2})) \\
		 =&\sum_{t\in\gf_{3^n}^*}\chi(\frac{1}{t}(\frac{1}{t}-1-u)(\frac{1}{t}-1+u)(\frac{1}{t}+1+\sqrt{1-u^2}))\\
		=&\sum_{t\in\gf_{3^n}^*}\chi((1-(1+u)t)(1-(1-u)t)(1+(1+\sqrt{1-u^2})t))\\
		=&-1+\sum_{t\in\gf_{3^n}}\chi((1-(1+u)t)(1-(1-u)t)(1+(1+\sqrt{1-u^2})t))\\
		 =&-1+\chi(1+u)\chi(1-u)\chi(1+\sqrt{1-u^2})\sum_{t\in\gf_{3^n}}\chi((t-\frac{1}{1+u})(t-\frac{1}{1-u})(t+\frac{1}{1+\sqrt{1-u^2}}))\\
		 =&-1+\chi(1+\sqrt{1-u^2})\sum_{t\in\gf_{3^n}}\chi(t(t-\frac{1}{1+\sqrt{1-u^2}}-\frac{1}{1+u})(t-\frac{1}{1+\sqrt{1-u^2}}-\frac{1}{1-u}))\\
		 =&-1+\chi(1+\sqrt{1-u^2})\sum_{t\in\gf_{3^n}}\chi(t(t^2+\frac{1-(1-u^2)\sqrt{1-u^2}}{u^2(1-u^2)}t+\frac{-u^4-u^2-1+\sqrt{1-u^2}}{u^4(1-u^2)}))\\
		=&-1+\chi(1+\sqrt{1-u^2})\chi(\frac{1-(1-u^2)\sqrt{1-u^2}}{u^2(1-u^2)})\\
		 &\cdot\sum_{t\in\gf_{3^n}}\chi(t(t^2+t+\frac{-u^4-u^2-1+\sqrt{1-u^2}}{u^4(1-u^2)}/(\frac{1-(1-u^2)\sqrt{1-u^2}}{u^2(1-u^2)})^2))\\
		=&-1+\sum_{t\in\gf_{3^n}}\chi(t(t^2+t+(\frac{\sqrt{1-u^2}}{1-\sqrt{1-u^2}})^3))\\
		=&-1+\sum_{t\in\gf_{3^n}}\chi(t(t^2+t+\frac{\sqrt{1-u^2}}{1-\sqrt{1-u^2}}))\\
		=&-1+\sum_{t\in\gf_{3^n}}\chi(t(t^2+t+\frac{1-u^2+\sqrt{1-u^2}}{u^2})).
	\end{align*}
	Then
	 \[\sum_{z\in\gf_{3^n}}\chi(g_1(z)g_2(z)g_3(z)g_4(z)g_5(z))=-\sum_{t\in\gf_{3^n}}\chi(t(t^2+t+\frac{1-u^2+\sqrt{1-u^2}}{u^2})).\]
	The desired result follows.
\end{proof}

\begin{lemma}When $u \in \mathcal{U}_0\backslash\gf_3$, we have
	\[ \sum_{z\in\gf_{3^n}}\chi(g_2(z)g_5(z))+\sum_{z\in\gf_{3^n}}\chi(g_3(z)g_4(z)g_5(z))=\chi(\sqrt{1-u^2}+1+u),\]
	and
		\[ \sum_{z\in\gf_{3^n}}\chi(g_3(z)g_5(z))+\sum_{z\in\gf_{3^n}}\chi(g_2(z)g_4(z)g_5(z))=\chi(\sqrt{1-u^2}+1-u).\]
\end{lemma}
\begin{proof} We only prove the first identity, as the proof of the second is very similar. 
	\[\sum_{z\in\gf_{3^n}}\chi(g_2(z)g_5(z))=-\chi(\varphi(u))\sum_{z\in\gf_{3^n}} \chi(z(z-1-u)(z+1-\sqrt{1-u^2})).\]
	Note that
	\begin{align*}
		&\sum_{z\in\gf_{3^n}} \chi(z(z-1-u)(z+1-\sqrt{1-u^2}))\\
		=&\sum_{z\in\gf_{3^n}} \chi(z(z^2-(u+\sqrt{1-u^2})z-(u+1)(1-\sqrt{1-u^2})))\\
		=&\sum_{z\in\gf_{3^n}^*} \chi(\frac{z^2-(u+\sqrt{1-u^2})z-(u+1)(1-\sqrt{1-u^2})}{z}).
	\end{align*}
	Let $t=\frac{z^2-(u+\sqrt{1-u^2})z-(u+1)(1-\sqrt{1-u^2})}{z}$. Then
	\[z^2-(t+u+\sqrt{1-u^2})z-(u+1)(1-\sqrt{1-u^2})=0,\]
	and $\Delta_t=(t+u+\sqrt{1-u^2})^2+(u+1)(1-\sqrt{1-u^2})=t^2-(u+\sqrt{1-u^2})t+(u-1)(1+\sqrt{1-u^2})$. \\
	We have
	\begin{align*}
		&\sum_{z\in\gf_{3^n}} \chi(z(z-1-u)(z+1-\sqrt{1-u^2}))\\
		=&\sum_{t\in\gf_{3^n}}\chi(t)(1+\chi(\Delta_t))\\
		=&\sum_{t\in\gf_{3^n}}\chi(t(t^2-(u+\sqrt{1-u^2})t+(u-1)(1+\sqrt{1-u^2})))\\
		=&-\sum_{t\in\gf_{3^n}}\chi(t(t^2+(u+\sqrt{1-u^2})t+(u-1)(1+\sqrt{1-u^2}))),
	\end{align*}
	then
\[\sum_{z\in\gf_{3^n}}\chi(g_2(z)g_5(z))=\chi(\varphi(u))\sum_{t\in\gf_{3^n}}\chi(t(t^2+(u+\sqrt{1-u^2})t+(u-1)(1+\sqrt{1-u^2}))).\]
	On the other hand,
	\begin{align*}
		 &\sum_{z\in\gf_{3^n}}\chi(g_3(z)g_4(z)g_5(z))\\
		 &=-\chi(\varphi(u))\sum_{z\in\gf_{3^n}}\chi(z(z-1+u)(z+1+\sqrt{1-u^2})(z+1-\sqrt{1-u^2})^2)\\
		&=-\chi(\varphi(u))\sum_{z\in\gf_{3^n},z\neq \sqrt{1-u^2}-1}\chi(z(z-1+u)(z+1+\sqrt{1-u^2}))\\
		&=-\chi(\varphi(u))(-\chi((\sqrt{1-u^2}-1)(\sqrt{1-u^2}+1+u)(-\sqrt{1-u^2}))\\
		&\quad+\sum_{z\in\gf_{3^n}}\chi(z(z-1+u)(z+1+\sqrt{1-u^2})))\\
		 &=\chi(\sqrt{1-u^2}+1+u)-\chi(\varphi(u))\sum_{z\in\gf_{3^n}}\chi(z(z^2+(u+\sqrt{1-u^2})z+(u-1)(1+\sqrt{1-u^2}))).
	\end{align*}
	Hence, the first identity ensues. 
\end{proof}

\section{On the number of solutions of the differential equation of $f_u$}\label{sec:DSpre}
Let $n$ be a positive odd integer, $u\in\gf_{3^n}$. Recall that the Ness-Helleseth function is defined as
\[f_u(x)=ux^{d_1}+x^{d_2},\]
where $d_1=\frac{3^n-1}{2}-1$ and $d_2=3^n-2$. To determine the differential spectrum of $f_u(x)$, attention should be given to the differential equation
\begin{equation}\label{DE}
	\begin{aligned}
		 \mathbb{D}_af_u(x)=u(x+a)^{\frac{3^n-1}{2}-1}+(x+a)^{3^n-2}-ux^{\frac{3^n-1}{2}-1}-x^{3^n-2}=b,
	\end{aligned}
\end{equation}
where $(a,b)\in\gf_{3^n}^*\times\gf_{3^n}$. This equation was studied in \cite{FurXia}. For the sake of completeness, we give some details here. 

We denote by $N(a,b)$, $N_1(a,b)$ and $N_2(a,b)$ the numbers of solutions of (\ref{DE}) in the sets $\gf_{3^n}$, $\{0,-a\}$ and $\gf_{3^n}\backslash\{0,-a\}$ respectively. Then $$N(a,b)=N_1(a,b)+N_2(a,b).$$
The following lemma is given in  \cite{FurXia}.
\begin{lemma}\cite[Lemma 3]{FurXia}\label{n1ab}
	The value of $N_1(a,b)$ is determined as follows:
	$$N_1(a,b)=
	\left\{
             \begin{array}{ll}
             2, & if\;b=a^{-1}\;and\;u=0, \\
             1, & if\;b=a^{-1}(1\pm u\chi(a))\;and\;u \ne 0, \\
             0, & otherwise.
             \end{array}
	\right.
	$$
	When $b = 0$, the value of $N_2(a,b)$ is given as
	$$N_2(a,0)=
	\left\{
             \begin{array}{ll}
             \frac{3^n-3}{4}, & if\;u\in\{\pm 1\}, \\
             0, & if\;u \in \mathcal{U}_0\backslash\{\pm1\}. \\
             \end{array}
	\right.
	$$
\end{lemma}
What needs to be calculated is $N_2(a,b)$ for $b\in\gf_{3^n}^*$. When $x\notin\{0, -a\}$, the differential equation is equivalent to
$$u(x+a)^{\frac{3^n-1}{2}}x+x-ux^{\frac{3^n-1}{2}}(x+a)-(x+a)=bx(x+a),$$
which can be simplified as
\begin{equation}
	\begin{aligned}
		bx^2+(ba-u(\tau_a-\tau_0))x+a(u\tau_0+1)=0, \label{QE}
	\end{aligned}
\end{equation}
where $\tau_a=\chi(x+a)$ and $\tau_0=\chi(x)$.
The discussion of the solutions of the quadratic equation above when $b\neq0$ has been clarified by Helleseth in \cite{ORG} and results are listed in Table \ref{tab:FCases}, in which $x_1$ and $x_2$ denote the two solutions of the quadratic equations in each case.
\begin{table}[!htp]
    \centering
    \caption{List of Equations and Solutions}
    \begin{tabular}{|c|c|c|c|c|}
        \hline
        Case  & I &  II &  III &  IV \\
        \hline
        $(\tau_a,\tau_0)$ & $(1,1)$ & $(1,-1)$ & $(-1,1)$ & $(-1,-1)$\\
        \hline
        $Equation$ & $bx^2+abx+a(u+1)=0$ & $bx^2+(u+ab)x-a(u-1)=0$ & $bx^2-(u-ab)x+a(u+1)=0$ & $bx^2+abx-a(u-1)=0$ \\
        \hline
		$x$ & $a\pm a\sqrt{1-\frac{u+1}{ab}}$ & $-\frac{1}{b}[-u-ab\pm \sqrt{u^2+a^2b^2-ab}]$ & $-\frac{1}{b}[u-ab\mp \sqrt{u^2+a^2b^2-ab}]$ & $a\pm a\sqrt{1+\frac{u-1}{ab}}$ \\
		\hline
		$x+a$ & $-a\pm a\sqrt{1-\frac{u+1}{ab}}$ & $-\frac{1}{b}[-u+ab\pm \sqrt{u^2+a^2b^2-ab}]$ & $-\frac{1}{b}[u+ab\mp \sqrt{u^2+a^2b^2-ab}]$ & $-a\pm a\sqrt{1+\frac{u-1}{ab}}$ \\
		\hline
		$x_1x_2$ & $\frac{a(u+1)}{b}$ & $-\frac{a(u-1)}{b}$ & $\frac{a(u+1)}{b}$ & $-\frac{a(u-1)}{b}$ \\
		\hline
		$x(x+a)$ & $-\frac{a(u+1)}{b}$ & $\frac{-u^2-ab\pm u\sqrt{u^2+a^2b^2-ab}}{b^2}$ & $\frac{-u^2-ab\mp u\sqrt{u^2+a^2b^2-ab}}{b^2}$ & $\frac{a(u-1)}{b}$ \\
		\hline
    \end{tabular}
	\label{tab:FCases}
 \end{table}

Drawing upon the information in Table \ref{tab:FCases}, the subsequent pivotal results have been unveiled. Note that the term \textit{a desired solution} refers to a solution of a certain quadratic equation in any case in Table \ref{tab:FCases} that indeed satisfies the corresponding condition on $(\tau_a,\tau_0)$. In the rest of this paper, we always assume that $u \in \mathcal{U}_0 \backslash \gf_3=\{u\in\gf_{3^n}\setminus \gf_3 | \chi(u+1)\neq\chi(u-1)\}$, then $\chi(1-u^2)=1$ and $(a,b) \in \gf_{3^n}^* \times \gf_{3^n}^*$. 
For the sake of brevity and clarity, for such fixed $u$ and $(a,b)$, we denote by $\mathrm{N_I}$ (respectively $\mathrm{N_{II}}$, $\mathrm{N_{III}}$ and $\mathrm{N_{IV}}$) the number of desired solutions in Case I (respectively, Case II, Case III and Case IV). Consequently, $N_2(a,b)=\mathrm{N_I}+\mathrm{N_{II}}+\mathrm{N_{III}}+\mathrm{N_{IV}}$. We discuss the values of $\mathrm{N_I}$, $\mathrm{N_{II}}$, $\mathrm{N_{III}}$ and $\mathrm{N_{IV}}$ as follows. It was proved in \cite{FurXia} that $\mathrm{N_I}\leq 1$ and $\mathrm{N_{IV}}\leq 1$. Moreover, the following proposition was proposed.
\begin{proposition}(\cite{FurXia})\label{n14} We have,
	\begin{enumerate}
		\item $\mathrm{N_I}=1$ if and only if $$\chi(1-\frac{u+1}{ab})=1\ and\ \chi(\frac{u+1}{ab})=-1.$$
		\item $\mathrm{N_{IV}}=1$ if and only if $$\chi(1+\frac{u-1}{ab})=1\ and\ \chi(\frac{u+1}{ab})=-1.$$
	\end{enumerate}
\end{proposition}%程序验证n=3,5时正确
Since $\chi(\frac{u+1}{ab})\neq 0$, the following corollary can be deduced immediately.
\begin{corollary} We have,
	\begin{enumerate}
		\item $\mathrm{N_I}=0$ if one of the subsequent three disjoint conditions is met:
		\begin{enumerate}
			\item $\chi(1-\frac{u+1}{ab})=0.$
			\item $\chi(1-\frac{u+1}{ab})=-1.$
			\item $\chi(1-\frac{u+1}{ab})=1$ and $\chi(\frac{u+1}{ab})=1.$
		\end{enumerate}
		\item $\mathrm{N_{IV}}=0$ if one of the subsequent three disjoint conditions is met:
		\begin{enumerate}
			\item $\chi(1+\frac{u-1}{ab})=0.$
			\item $\chi(1+\frac{u-1}{ab})=-1.$
			\item $\chi(1+\frac{u-1}{ab})=1$ and $\chi(\frac{u+1}{ab})=1.$
		\end{enumerate}
	\end{enumerate}
\end{corollary}%程序验证n=3,5时正确

As has been demonstrated in \cite{FurXia}, if $x$ is a solution of the quadratic equation in Case II, then $-(x+a)$ is a solution of the quadratic equation in Case III, and vice versa. Besides, $x$ and $-(x+a)$ cannot be desired solutions simultaneously. Therefore, it can be concluded that $\mathrm{N_{II}}+\mathrm{N_{III}}\leq 2$. More specifically, the following proposition showed the sufficient and necessary condition of  $\mathrm{N_{II}}+\mathrm{N_{III}}=2$.
\begin{proposition}(\cite{FurXia})We have,
	 $\mathrm{N_{II}}+\mathrm{N_{III}}=2$ if and only if
	$$\left\{
             \begin{array}{ll}
             \chi(u^2+a^2b^2-ab)=1, \\
             \chi(-u^2-ab-ab\sqrt{1-u^2})=1. \\
             \end{array}
	\right.
	$$
\end{proposition}
Next, we specifically consider the case when $\mathrm{N_{II}}+\mathrm{N_{III}}=1$. We have the following proposition.
\begin{proposition}We have,
	$\mathrm{N_{II}}+\mathrm{N_{III}}=1$ if and only if
		$$\left\{
	\begin{array}{ll}
		\chi(u^2+a^2b^2-ab)=0, \\
		\chi(a^2b^2-u^2)=1. \\
	\end{array}
	\right.
	$$
\end{proposition}
\begin{proof} The sufficiency is obvious. We only prove the necessity. When $\chi(u^2+a^2b^2-ab)=1$, the quadratic equation in Case II has two solutions, namely $x_1$ and $x_2$. Then the solutions of the quadratic equation in Case III are $-x_1-a$ and $-x_2-a$. Note that $x_1$ ($x_2$, respectively) is a desired solution if and only if $-x_2-a$ ($-x_1-a$, respectively) is a desired solution. Then $\mathrm{N_{II}}+\mathrm{N_{III}}\neq 1$ when  $\chi(u^2+a^2b^2-ab)=1$. Moreover, when  $\chi(u^2+a^2b^2-ab)=-1$, $\mathrm{N_{II}}+\mathrm{N_{III}}=0$. We conclude that if $\mathrm{N_{II}}+\mathrm{N_{III}}=1$, then $\chi(u^2+a^2b^2-ab)=0$.
	
When $\chi(u^2+a^2b^2-ab)=0$, let $x_0$ be the unique solution of the quadratic equation in Case II, then
$x_0=\frac{u+ab}{b}$. Moreover, the unique solution of the quadratic equation in Case III is $x'_0=-\frac{u-ab}{b}$. If $x_0$ is a desired solution, then $\chi(x_0)=\chi(\frac{u+ab}{b})=-1$ and $\chi(x_0+a)=\chi(\frac{u-ab}{b})=1$. If $x'_0$ is a desired solution, then $\chi(x'_0)=\chi(-\frac{u-ab}{b})=1$ and $\chi(x'_0+a)=\chi(-\frac{u+ab}{b})=-1$. Obviously, $x_0$ and $x'_0$ cannot be desired solutions simultaneously. If $\mathrm{N_{II}}+\mathrm{N_{III}}=1$, then $\chi(\frac{u+ab}{b})\chi(\frac{u-ab}{b})=-1$, i.e., $\chi(u^2-a^2b^2)=-1$. The proof is completed.
\end{proof}
Note that $\chi(-u^2-ab-ab\sqrt{1-u^2})=0$ implies that $\chi(u^2+a^2b^2-ab)=0$. Moreover, $\chi(u^2+a^2b^2-ab)=0$ and $\chi(a^2b^2-u^2)=0$ cannot hold simultaneously for  $u \in \mathcal{U}_0 \backslash \gf_3$. Then we have the following corollary.
\begin{corollary} $\mathrm{N_{II}}+\mathrm{N_{III}}=0$ if one of the subsequent three disjoint conditions is met:
	\begin{enumerate}
		\item $\chi(u^2+a^2b^2-ab)=-1.$
		\item $\chi(u^2+a^2b^2-ab)=0,\chi(a^2b^2-u^2)=-1.$
		\item $\chi(u^2+a^2b^2-ab)=1$ and $\chi(-u^2-ab-ab\sqrt{1-u^2})=-1.$
	\end{enumerate}
\end{corollary}

When $\chi(u^2+a^2b^2-ab)=0$, then $ab=-1\pm \sqrt{1-u^2}$, which implies that $\chi(\frac{u+1}{ab})=-\chi((u+1)\varphi(u))=1$. Then we conclude that $\mathrm{N_I}=\mathrm{N_{IV}}=0$ when $\mathrm{N_{II}}+\mathrm{N_{III}}=1$. Moreover, by propositions and corollaries demonstrated previously in this section, the discussion on the value of $N_2(a,b)=\mathrm{N_{I}}+\mathrm{N_{II}}+\mathrm{N_{III}}+\mathrm{N_{IV}}$ is finished. Recall that $N(a,b)=N_1(a,b)+N_2(a,b)$. By Lemma \ref{n1ab}, For $u\in \mathcal{U}_0 \backslash \gf_3$, $N_1(a,b)=1$ or $0$. When $N_1(a,b)=1$, then $ab=1\pm u$, the conditions in Proposition \ref{n14} cannot hold, hence $\mathrm{N_{I}}=\mathrm{N_{IV}}=0$. Moreover, for $u\in\mathcal{U}_0 \backslash \gf_3$, if $ab=1\pm u$, then $u^2+a^2b^2-ab\neq 0$. We conclude that $\mathrm{N_{II}}+\mathrm{N_{III}}\neq1$ when $N_1(a,b)=1$. We summarize the above discussion in the following Table \ref{tab:vN}.

\begin{table*}[!htp]
	\centering
	\caption{Values of $N(a,b)$}
	\label{table1}
	\begin{tabular}{|c|c|c|c|c|}
	\hline
	\multirow{2}{*}{$N(a,b)$}&\multirow{2}{*}{$N_1(a,b)$} & \multicolumn{3}{c|}{$N_2(a,b)$} \\
	\cline{3-5}
	
	 & & $\mathrm{N_I}$ & $\mathrm{N_{II}}+\mathrm{N_{III}}$ & $\mathrm{N_{IV}}$ \\
	\hline
	 $0$ & $0$ & $0$ & $0$ & $0$ \\
	\hline

	\multirow{4}{*}{$1$} & $1$ & $0$ & $0$ & $0$ \\
	\cline{2-5}

	\multirow{4}{*}{ } & $0$ & $1$ & $0$ & $0$ \\
	\cline{2-5}

	\multirow{4}{*}{ } & $0$ & $0$ & $1$ & $0$ \\
	\cline{2-5}

	\multirow{4}{*}{ } & $0$ & $0$ & $0$ & $1$ \\
	\hline

	\multirow{2}{*}{$2$} & $0$ & $0$ & $2$ & $0$ \\
	\cline{2-5}

	\multirow{2}{*}{ } & $0$ & $1$ & $0$ & $1$ \\
	\hline

	\multirow{3}{*}{$3$} & $1$ & $0$ & $2$ & $0$ \\
	\cline{2-5}

	\multirow{3}{*}{ } & $0$ & $1$ & $2$ & $0$ \\
	\cline{2-5}

	\multirow{3}{*}{ } & $0$ & $0$ & $2$ & $1$ \\
	\hline

	$4$ & $0$ & $1$ & $2$ & $1$ \\
	\hline
	\end{tabular}
	\label{tab:vN}
\end{table*}

By Table \ref{tab:vN}, we obtain the following sufficient and necessary conditions about the numbers of solutions of the differential equation (\ref{DE}). We mention that the sufficient and necessary condition for (\ref{DE}) to have 4 solutions was given in \cite{FurXia}.

\begin{proposition}\cite[Proposition 2]{FurXia}
	When $u \in \mathcal{U}_0 \backslash \gf_3$, the differential equation $\mathbb{D}_af_u(x)=b$ has four solutions if and only if $(a,b)$ satisfies the following conditions
	$$\left\{
             \begin{array}{ll}
             \chi(\frac{u+1}{ab})=-1, \\
             \chi(1-\frac{u+1}{ab})=1, \\
			 \chi(1-\frac{u-1}{ab})=1, \\
			 \chi(u^2+a^2b^2-ab)=1, \\
             \chi(-u^2-ab-ab\sqrt{1-u^2})=1. \\
             \end{array}
	\right.
	$$
	\label{prop:om4}
\end{proposition}
\begin{proposition}
	When $u \in \mathcal{U}_0 \backslash \gf_3$, the differential equation $\mathbb{D}_af_u(x)=b$ of the function $f_u(x)$ has three solutions if and only if $(a,b)$ satisfies one of the following conditions
	\begin{enumerate}
		\item $ab=1\pm u,\chi(u^2+a^2b^2-ab)=1, \chi(-u^2-ab-ab\sqrt{1-u^2})=1.$
		\item $\chi(1-\frac{u+1}{ab})=1, \chi(\frac{a(u+1)}{b})=-1, \chi(1+\frac{u-1}{ab})=-1, \chi(u^2+a^2b^2-ab)=1, \chi(-u^2-ab-ab\sqrt{1-u^2})=1.$
		\item $\chi(1-\frac{u+1}{ab})=-1, \chi(\frac{a(u+1)}{b})=-1, \chi(1+\frac{u-1}{ab})=1, \chi(u^2+a^2b^2-ab)=1, \chi(-u^2-ab-ab\sqrt{1-u^2})=1.$
	\end{enumerate}
	\label{prop:om3}
\end{proposition}%程序验证n=3,5时正确
\begin{proposition}
	When $u \in \mathcal{U}_0 \backslash \gf_3$, the differential equation $\mathbb{D}_af_u(x)=b$ of the function $f_u(x)$ has two solutions if and only if $(a,b)$ satisfies one of the following conditions
	\begin{enumerate}
		\item $\chi(1-\frac{u+1}{ab})=1, \chi(\frac{a(u+1)}{b})=-1, \chi(1+\frac{u-1}{ab})=1, \chi(u^2+a^2b^2-ab)=-1.$
		\item $\chi(1-\frac{u+1}{ab})=1, \chi(\frac{a(u+1)}{b})=-1, \chi(1+\frac{u-1}{ab})=1, \chi(u^2+a^2b^2-ab)=1,\chi(-u^2-ab-ab\sqrt{1-u^2})=-1.$
		\item $\chi(1-\frac{u+1}{ab})=-1,\chi(1+\frac{u-1}{ab})=-1,\chi(u^2+a^2b^2-ab)=1,\chi(-u^2-ab-ab\sqrt{1-u^2})=1.$
		\item $\chi(1-\frac{u+1}{ab})=-1,\chi(1+\frac{u-1}{ab})=1,\chi(\frac{a(u+1)}{b})=1,\chi(u^2+a^2b^2-ab)=1,\chi(-u^2-ab-ab\sqrt{1-u^2})=1.$
		\item $\chi(1-\frac{u+1}{ab})=1,\chi(\frac{a(u+1)}{b})=1,\chi(1+\frac{u-1}{ab})=-1,\chi(u^2+a^2b^2-ab)=1,\chi(-u^2-ab-ab\sqrt{1-u^2})=1.$
		\item $\chi(1-\frac{u+1}{ab})=1,\chi(\frac{a(u+1)}{b})=1,\chi(1+\frac{u-1}{ab})=1,\chi(u^2+a^2b^2-ab)=1,\chi(-u^2-ab-ab\sqrt{1-u^2})=1.$
	\end{enumerate}
	\label{prop:om2}
\end{proposition}%程序验证n=3,5时正确
\begin{proposition}
	When $u \in \mathcal{U}_0 \backslash \gf_3$, the differential equation $\mathbb{D}_af_u(x)=b$ of the function $f_u(x)$ has one solution if and only if $(a,b)$ satisfies one of the following conditions
	\begin{enumerate}
		\item $ab=1\pm u,\chi(u^2+a^2b^2-ab)=-1.$
		\item $ab=1\pm u,\chi(u^2+a^2b^2-ab)=1,\chi(-u^2-ab-ab\sqrt{1-u^2})=-1.$
		\item $\chi(u^2+a^2b^2-ab)=0,\chi(a^2b^2-u^2)=1.$
		\item $\chi(1-\frac{u+1}{ab})=1, \chi(\frac{a(u+1)}{b})=-1, \chi(1+\frac{u-1}{ab})=-1, \chi(u^2+a^2b^2-ab)=-1.$
		\item $\chi(1-\frac{u+1}{ab})=-1, \chi(\frac{a(u+1)}{b})=-1, \chi(1+\frac{u-1}{ab})=1, \chi(u^2+a^2b^2-ab)=-1.$
		\item $\chi(1-\frac{u+1}{ab})=1, \chi(\frac{a(u+1)}{b})=-1, \chi(1+\frac{u-1}{ab})=-1, \chi(u^2+a^2b^2-ab)=1, \chi(-u^2-ab-ab\sqrt{1-u^2})=-1.$
		\item $\chi(1-\frac{u+1}{ab})=-1, \chi(\frac{a(u+1)}{b})=-1, \chi(1+\frac{u-1}{ab})=1, \chi(u^2+a^2b^2-ab)=1, \chi(-u^2-ab-ab\sqrt{1-u^2})=-1.$
	\end{enumerate}
	\label{prop:om1}
\end{proposition}%程序验证n=3,5正确
\begin{proposition}
	When $u \in \mathcal{U}_0 \backslash \gf_3$, the differential equation $\mathbb{D}_af_u(x)=b$ of the function $f_u(x)$ has no solution if and only if $(a,b)$ satisfies one of the following conditions
	\begin{enumerate}
		\item $b=0.$
		\item $\chi(u^2+a^2b^2-ab)=0,\chi(a^2b^2-u^2)=-1.$
		\item $\chi(1-\frac{u+1}{ab})=-1, \chi(1+\frac{u-1}{ab})=-1, \chi(u^2+a^2b^2-ab)=-1.$
		\item $\chi(1-\frac{u+1}{ab})=-1, \chi(1+\frac{u-1}{ab})=1,\chi(\frac{a(u+1)}{b})=1,\chi(u^2+a^2b^2-ab)=-1.$
		\item $\chi(1-\frac{u+1}{ab})=1, \chi(\frac{a(u+1)}{b})=1, \chi(1+\frac{u-1}{ab})=-1, \chi(u^2+a^2b^2-ab)=-1.$
		\item $\chi(1-\frac{u+1}{ab})=1, \chi(\frac{a(u+1)}{b})=1, \chi(1+\frac{u-1}{ab})=1, \chi(u^2+a^2b^2-ab)=-1.$
		\item $\chi(1-\frac{u+1}{ab})=-1, \chi(1+\frac{u-1}{ab})=-1, \chi(u^2+a^2b^2-ab)=1, \chi(-u^2-ab-ab\sqrt{1-u^2})=-1.$
		\item $\chi(1-\frac{u+1}{ab})=-1, \chi(1+\frac{u-1}{ab})=1,\chi(\frac{a(u+1)}{b})=1,\chi(u^2+a^2b^2-ab)=1, \chi(-u^2-ab-ab\sqrt{1-u^2})=-1.$
		\item $\chi(1-\frac{u+1}{ab})=1, \chi(\frac{a(u+1)}{b})=1, \chi(1+\frac{u-1}{ab})=-1, \chi(u^2+a^2b^2-ab)=1,\chi(-u^2-ab-ab\sqrt{1-u^2})=-1.$
		\item $\chi(1-\frac{u+1}{ab})=1, \chi(\frac{a(u+1)}{b})=1, \chi(1+\frac{u-1}{ab})=1, \chi(u^2+a^2b^2-ab)=1,\chi(-u^2-ab-ab\sqrt{1-u^2})=-1.$
	\end{enumerate}
	\label{prop:om0}
\end{proposition}%程序验证n=3,5正确

\section{The Differential Spectrum of $f_u$ when $\chi(u+1)\neq \chi(u-1)$}\label{sec:DS}
Recall that $\omega_i=|\{(a,b)\in\gf_{p^n}^*\times\gf_{p^n}|\delta_F(a,b)=i\}|,0\leqslant i\leqslant\Delta_F$, where $\delta_F(a,b)$ denotes the number of solutions to the differential equation $\mathbb{D}_aF=b$. We are ready to investigate the differential spectrum of $f_u$.
As a prerequisite, we define two quadratic character sums, namely $\Gamma_3$ and $\Gamma_4$, as enumerated below.
\begin{align*}
	\Gamma_3 &=\sum_{z\in \gf_{3^n}}\chi(g_1(z)g_4(z))=-\chi(u+1)\sum\limits_{z\in\gf_{3^n}}\chi(z^3-z^2+u^2z). \\
	\Gamma_4 &=\sum_{z\in \gf_{3^n}}\chi(g_1(z)g_2(z)g_3(z)g_4(z))=-\chi(u+1)\sum\limits_{z\in\gf_{3^n}}\chi(z^5-(u^2+1)z^2+(u^2-u^4)z).
\end{align*}
These two character sums will be used in the differential spectrum of $f_u$. The main result of this paper is given as follows.
\begin{theorem}\label{thm:DS}
	Let $n\geq3$ be an odd integer and $f_u(x)=ux^{d_1}+x^{d_2}$ be the Ness-Helleseth function over $\gf_{3^n}$ with $d_1=\frac{3^n-1}{2}-1$ and $d_2=3^n-2$. Then, when $u \in \mathcal{U}_0\setminus \gf_3$, the differential spectrum of $f_u$ is given by
	\begin{align*}
		[
			\omega_0&=(3^n-1)(-1+\varepsilon+\frac{1}{32}(5\cdot3^{n+1}-17-\Gamma_4)),\\
			\omega_1&=(3^n-1)(3-\varepsilon+\frac{1}{16}(3^{n+1}+3+2\Gamma_3+\Gamma_4)),\\
			 \omega_2&=(3^n-1)(-\varepsilon+\frac{1}{4}(3^n-7-\Gamma_3)),\\
			\omega_3&=(3^n-1)(\varepsilon+\frac{1}{16}(3^n+1+2\Gamma_3-\Gamma_4)), \\
			\omega_4&=\frac{(3^n-1)}{32}(3^n+1+\Gamma_4)
		],
	\end{align*}
	where
	\begin{align}\label{epsilon}
	\varepsilon&=
	\left\{
            \begin{array}{ll}
            1, & \chi(u)=\chi(u+1),\chi((u+1)\sqrt{1-u^2}+(u-1)^2)=-1,or, \\
			   & \chi(u)=\chi(u-1),\chi((1-u)\sqrt{1-u^2}+(u+1)^2)=-1; \\
            0, & otherwise. \\
            \end{array}
	\right.
	\end{align}
	\label{thm:ds}
\end{theorem}

\begin{proof}
The proof of Theorem \ref{thm:ds} will be divided into five parts, where in each part $\omega_i$ (for $i\in\{0,1,2,3,4\}$) will be calculated.
\begin{enumerate}
\item\textit{Proof of $\omega_4$.} The sufficient and necessary condition for (\ref{DE}) to have 4 solutions was shown in Proposition \ref{prop:om4}. Let $ab=z$. For each $z\in\gf_{3^n}^*$, there are $3^n-1$ pairs of $(a,b)$ such that $ab=z$. Further we have, $$\omega_4=(3^n-1)n_4,$$ where $n_4$ denotes the number of $z$ satisfying the following system.
	$$\left\{
             \begin{array}{ll}
             \chi(g_1(z))=1, \\
             \chi(g_2(z))=1, \\
			 \chi(g_3(z))=1, \\
			 \chi(g_4(z))=1, \\
             \chi(g_5(z))=1, \\
             \end{array}
	\right.
	$$ where $g_i(i=1,2,3,4,5)$ are defined previously. Then by character sum $n_4$ can be expressed as
	\[n_4=\frac{1}{32}\sum_{z\in \gf_{3^n}\backslash\;A}(1+\chi(g_1(z)))(1+\chi(g_2(z)))(1+\chi(g_3(z)))(1+\chi(g_4(z)))(1+\chi(g_5(z))).\]
	By Table \ref{tab:vA},
	 \[\sum_{z\in A}(1+\chi(g_1(z)))(1+\chi(g_2(z)))(1+\chi(g_3(z)))(1+\chi(g_4(z)))(1+\chi(g_5(z)))=0.\]
	By the lemmas in Section II, it follows that
	\begin{align*}
	n_4&=\frac{1}{32}\sum_{z\in \gf_{3^n}}(1+\chi(g_1(z)))(1+\chi(g_2(z)))(1+\chi(g_3(z)))(1+\chi(g_4(z)))(1+\chi(g_5(z))) \\
	&=\frac{1}{32}(3^n+1+\sum_{z\in \gf_{3^n}}(g_1(z)g_2(z)g_3(z)g_4(z))-\chi(u+1)\chi(\varphi(u)-1)-\chi(\varphi(u))\chi(\varphi(u)-1)) \\
	&=\frac{1}{32}(3^n+1+\Gamma_4).
	\end{align*}
	The last identity holds since $\chi(u+1)\chi(\varphi(u))=-1$ and $\chi(u+1)+\chi(\varphi(u))=0$. Then the value of $\omega_4$ follows.
\item\textit{Proof of $\omega_3$.}
The sufficient and necessary condition for (\ref{DE}) to have 3 solutions was shown in Proposition \ref{prop:om3}. Let $ab=z$. For each $z\in\gf_{3^n}^*$, there are $3^n-1$ pairs of $(a,b)$ such that $ab=z$. Further we have $$\omega_3=(3^n-1)(n_{3,1}+n_{3,2}+n_{3,3}),$$ where the definitions of $n_{3,1}$ , $n_{3,2}$ and $n_{3,3}$ will be detailed below. \\
Let $n_{3,1}$ denote the number of $z$ satisfying
	$$\left\{
             \begin{array}{ll}
             z=1\pm u, \\
			 \chi(g_4(z))=1, \\
             \chi(g_5(z))=1. \\
             \end{array}
	\right.
	$$
Then we get
	$$n_{3,1}=
	\left\{
            \begin{array}{ll}
            1, & \chi(u)=\chi(u+1),\chi((u+1)\sqrt{1-u^2}+(u-1)^2)=-1,or, \\
			   & \chi(u)=\chi(u-1),\chi((1-u)\sqrt{1-u^2}+(u+1)^2)=-1; \\
            0, & otherwise. \\
            \end{array}
	\right.
	$$
Let $n_{3,2}$, $n_{3,3}$ denote the number of $z$ satisfying the following two equation systems respectively:
	$$\left\{
             \begin{array}{ll}
             \chi(g_1(z))=1, \\
             \chi(g_2(z))=1, \\
			 \chi(g_3(z))=-1, \\
			 \chi(g_4(z))=1, \\
             \chi(g_5(z))=1, \\
             \end{array}
	\right. \qquad
	\left\{
             \begin{array}{ll}
             \chi(g_1(z))=1, \\
             \chi(g_2(z))=-1, \\
			 \chi(g_3(z))=1, \\
			 \chi(g_4(z))=1, \\
             \chi(g_5(z))=1, \\
             \end{array}
	\right.
	$$ where $g_i(i=1,2,3,4,5)$ are defined previously. Then by character sum, $n_{3,1}$ can be expressed as
	\[32n_{3,2}=\sum_{z\in \gf_{3^n}\backslash\;A}(1+\chi(g_1(z)))(1+\chi(g_2(z)))(1-\chi(g_3(z)))(1+\chi(g_4(z)))(1+\chi(g_5(z))).\]
	By Table \ref{tab:vA},  \[\sum_{z\in A}(1+\chi(g_1(z)))(1+\chi(g_2(z)))(1-\chi(g_3(z)))(1+\chi(g_4(z)))(1+\chi(g_5(z)))=0.\]
	It follows that
	\[32n_{3,2}=\sum_{z\in \gf_{3^n}}(1+\chi(g_1(z)))(1+\chi(g_2(z)))(1-\chi(g_3(z)))(1+\chi(g_4(z)))(1+\chi(g_5(z))).\]
	Similarly, it can be concluded that
	\[32n_{3,3}=\sum_{z\in \gf_{3^n}}(1+\chi(g_1(z)))(1-\chi(g_2(z)))(1+\chi(g_3(z)))(1+\chi(g_4(z)))(1+\chi(g_5(z))).\]
	By utilizing the lemmas presented in Section II, the following sum can be derived
	\begin{align*}
	&n_{3,1}+n_{3,2}+n_{3,3} \\
	=&\varepsilon+\frac{1}{16}[3^n+1+2\sum_{z\in \gf_{3^n}}\chi(g_1(z)g_4(z))-\sum_{z\in \gf_{3^n}}\chi(g_1(z)g_2(z)g_3(z)g_4(z))\\
	&-\chi(\varphi(u))\chi(\varphi(u)-1)-\chi(u+1)\chi(\varphi(u)-1)] \\
	 =&\varepsilon+\frac{1}{16}[3^n+1+2\Gamma_3-\Gamma_4-\chi(\varphi(u))\chi(\varphi(u)-1)-\chi(u+1)\chi(\varphi(u)-1)] \\
	=&\varepsilon+\frac{1}{16}(3^n+1+2\Gamma_3-\Gamma_4),
	\end{align*}
	where $\epsilon$ was defined in (\ref{epsilon}).
	
%	$$\varepsilon=
%	\left\{
 %           \begin{array}{ll}
 %           1, & \chi(u)=\chi(u+1),\chi((u+1)\sqrt{1-u^2}+(u-1)^2)=-1,or, \\
%			   & \chi(u)=\chi(u-1),\chi((1-u)\sqrt{1-u^2}+(u+1)^2)=-1; \\
 %           0, & otherwise. \\
%            \end{array}
%	\right.
%	$$
\item\textit{Proof of $\omega_2$.}
	The sufficient and necessary condition for (\ref{DE}) to have 2 solutions was shown in Proposition \ref{prop:om2}.  Let $ab=z$. For each $z\in\gf_{3^n}^*$, there are $3^n-1$ pairs of $(a,b)$ such that $ab=z$. Further we have $$\omega_2=(3^n-1)(n_{2,1}+n_{2,2}+n_{2,3}+n_{2,4}+n_{2,5}+n_{2,6}),$$ where $n_{2,1},n_{2,2},n_{2,3},n_{2,4},n_{2,5},n_{2,6}$ denote the number of $z$ satisfying the following six equation systems respectively:
	$$\left\{
             \begin{array}{ll}
             \chi(g_1(z))=1, \\
             \chi(g_2(z))=1, \\
			 \chi(g_3(z))=1, \\
			 \chi(g_4(z))=-1, \\
             \end{array}
	\right. \qquad
	\left\{
             \begin{array}{ll}
             \chi(g_2(z))=-1, \\
			 \chi(g_3(z))=-1, \\
			 \chi(g_4(z))=1, \\
			 \chi(g_5(z))=1, \\
             \end{array}
	\right. \qquad
	\left\{
             \begin{array}{ll}
			 \chi(g_1(z))=1, \\
             \chi(g_2(z))=1, \\
			 \chi(g_3(z))=1, \\
			 \chi(g_4(z))=1, \\
			 \chi(g_5(z))=-1, \\
             \end{array}
	\right.$$$$
	\left\{
             \begin{array}{ll}
			 \chi(g_1(z))=-1, \\
             \chi(g_2(z))=-1, \\
			 \chi(g_3(z))=1, \\
			 \chi(g_4(z))=1, \\
			 \chi(g_5(z))=1, \\
             \end{array}
	\right. \qquad
	\left\{
             \begin{array}{ll}
			 \chi(g_1(z))=-1, \\
             \chi(g_2(z))=1, \\
			 \chi(g_3(z))=-1, \\
			 \chi(g_4(z))=1, \\
			 \chi(g_5(z))=1, \\
             \end{array}
	\right. \qquad
	\left\{
             \begin{array}{ll}
			 \chi(g_1(z))=-1, \\
             \chi(g_2(z))=1, \\
			 \chi(g_3(z))=1, \\
			 \chi(g_4(z))=1, \\
			 \chi(g_5(z))=1, \\
             \end{array}
	\right.
	$$ where $g_i(i=1,2,3,4,5)$ are defined previously.
	Then by character sum, $n_{21}$ can be expressed as
	\[16n_{2,1}=\sum_{z\in \gf_{3^n}\backslash\;A}(1+\chi(g_1(z)))(1+\chi(g_2(z)))(1+\chi(g_3(z)))(1-\chi(g_4(z))).\]
	By Table \ref{tab:vA}, \[\sum_{z\in A}(1+\chi(g_1(z)))(1+\chi(g_2(z)))(1+\chi(g_3(z)))(1-\chi(g_4(z)))=0.\]
	It follows that
	\[16n_{2,1}=\sum_{z\in \gf_{3^n}}(1+\chi(g_1(z)))(1+\chi(g_2(z)))(1+\chi(g_3(z)))(1-\chi(g_4(z))).\]
	Similarly, it can be concluded that
	\begin{align*}
	16n_{2,2}&=\sum_{z\in \gf_{3^n}\backslash\,\{1\pm u,-1\pm \sqrt{1-u^2}\}}(1-\chi(g_2(z)))(1-\chi(g_3(z)))(1+\chi(g_4(z)))(1+\chi(g_5(z))).\\
	32n_{2,3}&=\sum_{z\in \gf_{3^n}}(1+\chi(g_1(z)))(1+\chi(g_2(z)))(1+\chi(g_3(z)))(1+\chi(g_4(z)))(1-\chi(g_5(z)))-4.\\
	32n_{2,4}&=\sum_{z\in \gf_{3^n}\backslash\,A}(1-\chi(g_1(z)))(1-\chi(g_2(z)))(1+\chi(g_3(z)))(1+\chi(g_4(z)))(1+\chi(g_5(z))).\\
	32n_{2,5}&=\sum_{z\in \gf_{3^n}\backslash\,A}(1-\chi(g_1(z)))(1+\chi(g_2(z)))(1-\chi(g_3(z)))(1+\chi(g_4(z)))(1+\chi(g_5(z))).\\
	32n_{2,6}&=\sum_{z\in \gf_{3^n}\backslash\,A}(1-\chi(g_1(z)))(1+\chi(g_2(z)))(1+\chi(g_3(z)))(1+\chi(g_4(z)))(1+\chi(g_5(z))).
	\end{align*}
	By utilizing the lemmas presented in Section II, the following sum can be derived
	\begin{align*}
	&n_{2,1}+n_{2,2}+n_{2,3}+n_{2,4}+n_{2,5}+n_{2,6} \\
	=&\frac{1}{8}[2\cdot3^n-14-2\sum_{z\in \gf_{3^n}}\chi(g_1(z)g_4(z))-\chi(\varphi(u))\chi(\varphi(u)-1)+\chi(u+1)\chi(\varphi(u)-1) \\
	&-2(1+\chi(u-u^2))(1-\chi((u+1)(\sqrt{1-u^2})+(u-1)^2))-2\chi(u^2-1-\sqrt{1-u^2}) \\
	&-2(1-\chi(u^2+u))(1-\chi((1-u)(\sqrt{1-u^2})+(u+1)^2))] \\
	 =&\frac{1}{8}[2\cdot3^n-14-2\Gamma_3-\chi(\varphi(u))\chi(\varphi(u)-1)+\chi(u+1)\chi(\varphi(u)-1) \\
	&-2(1+\chi(u-u^2))(1-\chi((u+1)(\sqrt{1-u^2})+(u-1)^2))-2\chi(u^2-1-\sqrt{1-u^2}) \\
	&-2(1-\chi(u^2+u))(1-\chi((1-u)(\sqrt{1-u^2})+(u+1)^2))] \\
	 =&\frac{1}{4}(3^n-7-\Gamma_3+\chi(u+1)-(1-\chi(u^2+u))(1-\chi((1-u)(\sqrt{1-u^2})+(u+1)^2))\\
	&-(1+\chi(u-u^2))(1-\chi((u+1)(\sqrt{1-u^2})+(u-1)^2))-\chi(u^2-1-\sqrt{1-u^2}))\\
	=&\frac{1}{4}(3^n-7-4\varepsilon-\Gamma_3+\chi(u+1)-\chi(u^2-1-\sqrt{1-u^2}))\\
	=&-\varepsilon+\frac{1}{4}(3^n-7-\Gamma_3),
	\end{align*}
	%\sum_{z\in A^*}8(1+\chi(g_4))(1+\chi(g_5))
	where $\varepsilon$ has been defined in (\ref{epsilon}). The last identity holds since\\ $\chi(u^2-1-\sqrt{1-u^2})=\chi(\sqrt{1-u^2}(-\sqrt{1-u^2}-1))=-\chi(\varphi(u))$ and $\chi(u+1)\chi(\varphi(u))=-1$. 
\item\textit{Proof of $\omega_1$.}
	The sufficient and necessary condition for (\ref{DE}) to have 1 solution was shown in Proposition \ref{prop:om1}. Let $ab=z$. For each $z\in\gf_{3^n}^*$, there are $3^n-1$ pairs of $(a,b)$ such that $ab=z$. Further we have $$\omega_1=(3^n-1)(n_{1,1}+n_{1,2}+n_{1,3}+n_{1,4}+n_{1,5}+n_{1,6}+n_{1,7}),$$ where the definitions of $n_{1,1}$, $n_{1,2}$, $n_{1,3}$, $n_{1,4}$, $n_{1,5}$, $n_{1,6}$ and $n_{1,7}$ will be detailed below. \\
	Let $n_{1,1}$, $n_{1,2}$, $n_{1,3}$ denote the number of $z$ satisfying the following two equation systems respectively:
	$$\left\{
            \begin{array}{ll}
            z=1\pm u, \\
			\chi(g_4(z))=-1, \\
            \end{array}
	\right. \qquad
	\left\{
            \begin{array}{ll}
			z=1\pm u, \\
			\chi(g_4(z))=1, \\
			\chi(g_5(z))=-1, \\
            \end{array}
	\right. \qquad
	\left\{
            \begin{array}{ll}
			\chi(g_4(z))=0, \\
			\chi(g_5(z))=1. \\
            \end{array}
	\right.
	$$
Then we get
\begin{align*}
	n_{1,1}&=
	\left\{
            \begin{array}{ll}
            1, &\chi(u)=\chi(u-1)~or~\chi(u)=\chi(u+1); \\
			0, & otherwise. \\
            \end{array}
	\right.\\
	n_{1,2}&=
	\left\{
            \begin{array}{ll}
			1, & \chi(u)=\chi(u+1),\chi((u+1)\sqrt{1-u^2}+(u-1)^2)=1,or \\
			   & \chi(u)=\chi(u-1),\chi((1-u)\sqrt{1-u^2}+(u+1)^2)=1; \\
			0, & otherwise. \\
            \end{array}
	\right.\\
	n_{1,3}&=
	\left\{
            \begin{array}{ll}
			1, & \chi(u^2-1+\sqrt{1-u^2})=1~or~\chi(u^2-1-\sqrt{1-u^2})=1; \\
			0, & otherwise. \\
            \end{array}
	\right.
	\end{align*}
	Note that either of $\chi(u)=\chi(u-1)$ or $\chi(u)=\chi(u+1)$ must hold since $\chi(u-1)\neq\chi(u+1)$. Similarly, $\chi(u^2-1+\sqrt{1-u^2})\neq\chi(u^2-1-\sqrt{1-u^2})$ since $(u^2-1+\sqrt{1-u^2})(u^2-1-\sqrt{1-u^2})=u^2(u^2-1)$, which is a nonsquare. It follows that either of $\chi(u^2-1+\sqrt{1-u^2})=1$ or $\chi(u^2-1-\sqrt{1-u^2})=1$ must hold since $\chi(u^2-1+\sqrt{1-u^2})\neq(u^2-1-\sqrt{1-u^2})$ and neither of them could be $0$. Then we can conclude that $n_{1,1}=1$ and $n_{1,3}=1$. \\
	Let $n_{1,4}$, $n_{1,5}$, $n_{1,6}$, $n_{1,7}$ denote the number of $z$ satisfying the following four equation systems respectively:
	$$\left\{
             \begin{array}{ll}
             \chi(g_1(z))=1, \\
             \chi(g_2(z))=1, \\
			 \chi(g_3(z))=-1, \\
			 \chi(g_4(z))=-1, \\
             \end{array}
	\right. \qquad
	\left\{
             \begin{array}{ll}
             \chi(g_1(z))=1, \\
             \chi(g_2(z))=-1, \\
			 \chi(g_3(z))=1, \\
			 \chi(g_4(z))=-1, \\
             \end{array}
	\right. \qquad
	\left\{
             \begin{array}{ll}
             \chi(g_1(z))=1, \\
             \chi(g_2(z))=1, \\
			 \chi(g_3(z))=-1, \\
			 \chi(g_4(z))=1, \\
             \chi(g_5(z))=-1, \\
             \end{array}
	\right.  \qquad
	\left\{
             \begin{array}{ll}
             \chi(g_1(z))=1, \\
             \chi(g_2(z))=-1, \\
			 \chi(g_3(z))=1, \\
			 \chi(g_4(z))=1, \\
             \chi(g_5(z))=-1, \\
             \end{array}
	\right. $$ where $g_i(i=1,2,3,4,5)$ are defined previously.
	Then by character sum, $n_{1,3}$ can be expressed as
	\[16n_{1,4}=\sum_{z\in \gf_{3^n}\backslash\;A}(1+\chi(g_1(z)))(1+\chi(g_2(z)))(1-\chi(g_3(z)))(1-\chi(g_4(z))).\]
	By Table \ref{tab:vA}, \[\sum_{z\in A}(1+\chi(g_1(z)))(1+\chi(g_2(z)))(1-\chi(g_3(z)))(1-\chi(g_4(z)))=0.\]
	It follows that
	\[16n_{1,4}=\sum_{z\in \gf_{3^n}}(1+\chi(g_1(z)))(1+\chi(g_2(z)))(1-\chi(g_3(z)))(1-\chi(g_4(z))).\]
	Similarly, it can be concluded that
	\begin{align*}
	16n_{1,5}&=\sum_{z\in \gf_{3^n}}(1+\chi(g_1(z)))(1-\chi(g_2(z)))(1+\chi(g_3(z)))(1-\chi(g_4(z))).\\
	 32n_{1,6}&=\sum\limits_{z\in\gf_{3^n}}(1+\chi(g_1(z)))(1+\chi(g_2(z)))(1-\chi(g_3(z)))(1+\chi(g_4(z)))(1-\chi(g_5(z)))-4.\\
	 32n_{1,7}&=\sum\limits_{z\in\gf_{3^n}}(1+\chi(g_1(z)))(1-\chi(g_2(z)))(1+\chi(g_3(z)))(1+\chi(g_4(z)))(1-\chi(g_5(z)))-4.
	\end{align*}
	By utilizing the lemmas presented in Section II, the following sum can be derived
	\begin{align*}
	&n_{1,1}+n_{1,2}+n_{1,3}+n_{1,4}+n_{1,5}+n_{1,6}+n_{1,7} \\
	 =&2+(1-\varepsilon)+\frac{1}{16}[3\cdot3^n+3+2\sum\limits_{z\in\gf_{3^n}}\chi(g_1(z)g_4(z))+\sum\limits_{z\in\gf_{3^n}}\chi(g_1(z)g_2(z)g_3(z)g_4(z))\\
	 & +\chi(\varphi(u))\chi(\varphi(u)-1)+\chi(u+1)\chi(\varphi(u)-1)]\\
	 =&3-\varepsilon+\frac{1}{16}[3\cdot3^n+3+2\Gamma_3+\Gamma_4+\chi(\varphi(u))\chi(\varphi(u)-1)+\chi(u+1)\chi(\varphi(u)-1)] \\
	=&3-\varepsilon+\frac{1}{16}(3^{n+1}+3+2\Gamma_3+\Gamma_4),
	\end{align*}
	where $\varepsilon$ has been defined in (\ref{epsilon}).
\item\textit{Proof of $\omega_0$.}
	The sufficient and necessary condition for (\ref{DE}) to have no solution was shown in Proposition \ref{prop:om0}. Let $ab=z$. For each $z\in\gf_{3^n}^*$, there are $3^n-1$ pairs of $(a,b)$ such that $ab=z$. Further we have $$\omega_0=(3^n-1)(n_{0,1}+n_{0,2}+n_{0,3}+n_{0,4}+n_{0,5}+n_{0,6}+n_{0,7}+n_{0,8}+n_{0,9}+n_{0,10}),$$ where $n_{0,1}=1$ for the condition $z=0$ and the definitions of $n_{0,2}$, $n_{0,3}$, $n_{0,4}$, $n_{0,5}$, $n_{0,6}$, $n_{0,7}$, $n_{0,8}$, $n_{0,9}$ and $n_{0,10}$ will be detailed below. \\
	Let $n_{0,2}$ denote the number of $z$ satisfying
	$$\left\{
             \begin{array}{ll}
			 \chi(g_4(z))=0, \\
			 \chi(z^2-u^2)=-1, \\
             \end{array}
	\right.$$
	then
	$$n_{0,2}=
	\left\{
             \begin{array}{ll}
			 1, & \chi(u^2-1+\sqrt{1-u^2})=-1~or~\chi(u^2-1-\sqrt{1-u^2})=-1; \\
			 0, & otherwise. \\
             \end{array}
	\right.$$
	Note that $\chi(u^2-1+\sqrt{1-u^2})\neq\chi(u^2-1-\sqrt{1-u^2})$ since $(u^2-1+\sqrt{1-u^2})(u^2-1-\sqrt{1-u^2})=u^2(u^2-1)$, which is a nonsquare. It follows that either of $\chi(u^2-1+\sqrt{1-u^2})=-1$ or $\chi(u^2-1-\sqrt{1-u^2})=-1$ must hold. Then we can conclude that $n_{0,2}=1$. \\
	Let $n_{0,3}$, $n_{0,4}$, $n_{0,5}$, $n_{0,6}$, $n_{0,7}$, $n_{0,8}$, $n_{0,9}$, $n_{0,10}$ denote the number of $z$ satisfying the following eight equation systems respectively:
	$$\left\{
             \begin{array}{ll}
             \chi(g_2(z))=-1, \\
			 \chi(g_3(z))=-1, \\
			 \chi(g_4(z))=-1, \\
             \end{array}
	\right. \qquad
	\left\{
             \begin{array}{ll}
			 \chi(g_1(z))=-1, \\
             \chi(g_2(z))=-1, \\
			 \chi(g_3(z))=1, \\
			 \chi(g_4(z))=-1, \\
             \end{array}
	\right. \qquad
	\left\{
             \begin{array}{ll}
			 \chi(g_1(z))=-1, \\
             \chi(g_2(z))=1, \\
			 \chi(g_3(z))=-1, \\
			 \chi(g_4(z))=-1, \\
             \end{array}
	\right. \qquad
	\left\{
             \begin{array}{ll}
			 \chi(g_1(z))=-1, \\
             \chi(g_2(z))=1, \\
			 \chi(g_3(z))=1, \\
			 \chi(g_4(z))=-1, \\
             \end{array}
	\right.
	$$$$
	\left\{
             \begin{array}{ll}
             \chi(g_2(z))=-1, \\
			 \chi(g_3(z))=-1, \\
			 \chi(g_4(z))=1, \\
			 \chi(g_5(z))=-1, \\
             \end{array}
	\right. \qquad
	\left\{
             \begin{array}{ll}
			 \chi(g_1(z))=-1, \\
             \chi(g_2(z))=-1, \\
			 \chi(g_3(z))=1, \\
			 \chi(g_4(z))=1, \\
			 \chi(g_5(z))=-1, \\
             \end{array}
	\right. \qquad
	\left\{
             \begin{array}{ll}
			 \chi(g_1(z))=-1, \\
             \chi(g_2(z))=1, \\
			 \chi(g_3(z))=-1, \\
			 \chi(g_4(z))=1, \\
			 \chi(g_5(z))=-1, \\
             \end{array}
	\right. \qquad
	\left\{
             \begin{array}{ll}
			 \chi(g_1(z))=-1, \\
             \chi(g_2(z))=1, \\
			 \chi(g_3(z))=1, \\
			 \chi(g_4(z))=1, \\
			 \chi(g_5(z))=-1, \\
             \end{array}
	\right.
	$$ where $g_i(i=1,2,3,4,5)$ are defined previously.
	Then by character sum, $n_{0,3}$ can be expressed as
	\[8n_{0,3}=\sum_{z\in \gf_{3^n}\backslash\;A}(1-\chi(g_2(z)))(1-\chi(g_3))(1-\chi(g_4(z))).\]
	Similarly, it can be concluded that
	\begin{align*}
	16n_{0,4}&=\sum_{z\in \gf_{3^n}\backslash\,A}(1-\chi(g_1(z)))(1-\chi(g_2(z)))(1+\chi(g_3(z)))(1-\chi(g_4(z))).\\
	16n_{0,5}&=\sum_{z\in \gf_{3^n}\backslash\,A}(1-\chi(g_1(z)))(1+\chi(g_2(z)))(1-\chi(g_3(z)))(1-\chi(g_4(z))).\\
	16n_{0,6}&=\sum_{z\in \gf_{3^n}\backslash\,A}(1-\chi(g_1(z)))(1+\chi(g_2(z)))(1+\chi(g_3(z)))(1-\chi(g_4(z))).\\
	16n_{0,7}&=\sum_{z\in \gf_{3^n}\backslash\,A}(1-\chi(g_2(z)))(1-\chi(g_3(z)))(1+\chi(g_4(z)))(1-\chi(g_5(z))).\\
	32n_{0,8}&=\sum_{z\in \gf_{3^n}\backslash\,A}(1-\chi(g_1(z)))(1-\chi(g_2(z)))(1+\chi(g_3(z)))(1+\chi(g_4(z)))(1-\chi(g_5(z))).\\
	32n_{0,9}&=\sum_{z\in \gf_{3^n}\backslash\,A}(1-\chi(g_1(z)))(1+\chi(g_2(z)))(1-\chi(g_3(z)))(1+\chi(g_4(z)))(1-\chi(g_5(z))).\\
	32n_{0,10}&=\sum_{z\in \gf_{3^n}\backslash\,A}(1-\chi(g_1(z)))(1+\chi(g_2(z)))(1+\chi(g_3(z)))(1+\chi(g_4(z)))(1-\chi(g_5(z))).
	\end{align*}
	By utilizing Table \ref{tab:vA} and the lemmas  presented in Section II, the following sum can be derived
	\begin{align*}
	&n_{0,1}+n_{0,2}+n_{0,3}+n_{0,4}+n_{0,5}+n_{0,6}+n_{0,7}+n_{0,8}+n_{0,9}+n_{0,10}\\
	=&2+\frac{1}{32}[15\cdot3^n-81-\sum_{z\in \gf_{3^n}}\chi(g_1(z)g_2(z)g_3(z)g_4(z))+5\chi(\varphi(u))\chi(\varphi(u)-1)-3\chi(u+1)\chi(\varphi-1)\\
	&+8\chi(u^2-1-\sqrt{1-u^2})-8(1+\chi(u-u^2))(1+\chi((u+1)(\sqrt{1-u^2})+(u-1)^2)) \\
	&-8(1-\chi(u^2+u))(1+\chi((1-u)(\sqrt{1-u^2})+(u+1)^2))] \\
	 =&2+\frac{1}{32}[15\cdot3^n-81-\Gamma_4+5\chi(\varphi(u))\chi(\varphi(u)-1)-3\chi(u+1)\chi(\varphi-1)\\
	&+8\chi(u^2-1-\sqrt{1-u^2})-8(1+\chi(u-u^2))(1+\chi((u+1)(\sqrt{1-u^2})+(u-1)^2)) \\
	&-8(1-\chi(u^2+u))(1+\chi((1-u)(\sqrt{1-u^2})+(u+1)^2))] \\
	 =&(2+\frac{1}{32}(15\cdot3^n-81-\Gamma_4-8\chi(u+1)-8(1-\chi(u^2+u))(1+\chi((1-u)(\sqrt{1-u^2})+(u+1)^2))\\
	&+8\chi(u^2-1-\sqrt{1-u^2})-8(1+\chi(u-u^2))(1+\chi((u+1)(\sqrt{1-u^2})+(u-1)^2))))\\
	 =&(2+\frac{1}{32}(5\cdot3^{n+1}-81-32(1-\varepsilon)-\Gamma_4-8\chi(u+1)+8\chi(u^2-1-\sqrt{1-u^2}))) \\
	 =&-1+\varepsilon+\frac{1}{32}(5\cdot3^{n+1}-17-\Gamma_4),
	\end{align*}
	where $\varepsilon$ has beem defined in (\ref{epsilon}).
\end{enumerate}
This completes the proof of Theorem \ref{thm:DS}.
\end{proof}

\begin{remark}
	Recall that the elements $\omega_i (i=0,1,2,3,4)$ satisfy two identities  in (\ref{eqt:omega}). Namely,
		\begin{align*}
		\left\{
		\begin{array}{ccc}
			\omega_0+\omega_1+\omega_2+\omega_3+\omega_4&=&(3^n-1)3^n, \\
			\omega_1+2\omega_2+3\omega_3+4\omega_4&=&(3^n-1)3^n.
		\end{array}
		\right. 
	\end{align*}
	After the values of $\omega_4$, $\omega_3$ and $\omega_2$ are determined, $\omega_1$ and $\omega_0$ can be deduced by solving the above system.
\end{remark}

\begin{example} Let $p=3$, $n=3$ and $u=w^4$, where $w$ is a primitive element in $\gf_{3^n}^*$. Then $u\in\mathcal{U}_0\backslash\gf_3$, $\varepsilon=0$, $\Gamma_3=-4$ and $\Gamma_4=4$. 
	By Theorem \ref{thm:DS}, the differential spectrum of $f_u$ is 
		$$\mathbb{S}=[\omega_0=286, \omega_1=208, \omega_2=156, \omega_3=26, \omega_4=26],$$
	which coincides with the result calculated directly by MAGMA.
\end{example}

\begin{example} Let $p=3$, $n=5$ and $u=w^{210}$, where $w$ is a primitive element in $\gf_{3^n}^*$. Then $u\in\mathcal{U}_0\backslash\gf_3$, $\varepsilon=1$,$\Gamma_3=-4$ and $\Gamma_4=12$.
	By Theorem \ref{thm:DS}, the differential spectrum of $f_u$ is 
   $$\mathbb{S}=[\omega_0=27346, \omega_1=11616, \omega_2=14278, \omega_3=3630, \omega_4=1936].$$
   which coincides with the result calculated directly by MAGMA.
\end{example}

\begin{example} Let $p=3$, $n=7$ and $u=w$, where $w$ is a primitive element in $\gf_{3^n}^*$. Then $u\in\mathcal{U}_0\backslash\gf_3$, $\varepsilon=1$,$\Gamma_3=-28$ and $\Gamma_4=-12$.
	By Theorem \ref{thm:DS}, the differential spectrum of $f_u$ is 
   $$\mathbb{S}=[\omega_0=2240650, \omega_1=891888, \omega_2=1204486, \omega_3=295110, \omega_4=148648].$$
   which coincides with the result calculated directly by MAGMA.
\end{example}

\section{Concluding remarks}\label{sec:con}
In this paper, we conducted an in-depth investigation of the differential properties of the Ness-Helleseth function. For $u\in\mathcal{U}_0\setminus\gf_3$, we expressed the differential spectrum in terms of quadratic character sums. This completed the work on the differential properties of the Ness-Helleseth function. Besides, we obtained a series of identities of character sums, which may be used in some other areas. It may be interesting to consider applications of the differential spectrum of the Ness-Helleseth function in other areas such as sequence design, coding theory and combinatorial design. Moreover, the study of the Ness-Helleseth function can be extended to $p>3$ \cite{Lyu2024AFS}, \cite{xia2024investigation}, and the investigation of the differential spectrum of such function will be undertaken in our further work.
	%\bibliographystyle{plain}
	%\bibliography{2dsref.bib}

\bibliographystyle{IEEEtranS}

\bibliography{RXY-1}

\end{document}